\documentclass[11pt]{article}
 
\usepackage[left=1in,right=1in,top=1in,bottom=1in]{geometry}
\usepackage{amssymb,amsthm,amsmath}
\usepackage{braket}
\usepackage{xcolor}
\usepackage{algorithm}
\usepackage{graphicx}
\usepackage[noend]{algpseudocode}
\usepackage[utf8]{inputenc}
 
\usepackage[hidelinks, final=true,colorlinks = false,allcolors = {blue}]{hyperref}
 
\newtheorem{theorem}{Theorem}
\newtheorem{proposition}{Proposition}
\newtheorem{lemma}{Lemma}
\newtheorem{corollary}{Corollary}
\newtheorem{definition}{Definition}
\newtheorem{fact}{Fact}
\newtheorem{claim}{Claim}

\newcommand{\tR}{\widetilde{R}}
\newcommand{\val}{\mathrm{val}}

\DeclareMathOperator{\Exp}{\mathbb{E}}
\DeclareMathOperator{\poly}{\mathrm{poly}}

\newcommand{\tO}{\widetilde{O}}

\newcommand{\complex}{\mathbb{C}}
\newcommand{\real}{\mathbb{R}}

\renewcommand{\natural}{\mathbb{N}}
 
\newcommand{\mincut}{\textsc{min\;cut}}
\newcommand{\minstcut}{\textsc{min} $st$-\textsc{cut}}
\newcommand{\maxcut}{\textsc{max\;cut}}
\newcommand{\spcut}{\textsc{sparsest\;cut}}
\newcommand{\balsep}{\textsc{balanced\;separator}}
\newcommand{\minfind}{\textsc{minfind}}
\newcommand{\qspars}{{\normalfont\textbf{Quantum\,Sparsify}}}
\newcommand{\halfspars}{{\normalfont\textbf{Half\,Sparsify}}}
\newcommand{\spanner}{{\normalfont\textbf{Spanner}}}
\newcommand{\cost}{\mathrm{cost}}
\newcommand{\dist}{\mathrm{dist}}
\newcommand{\spt}{{\normalfont\textbf{SPT}}}
\newcommand{\tcov}{\tau_{\mathrm{cov}}}
\newcommand{\findb}{\mathrm{FindBits}}
\newcommand{\ORf}{\mathrm{OR}}
 
\begin{document}
 
\title{Quantum Speedup for Graph Sparsification, Cut Approximation and Laplacian Solving}
\author{Simon Apers\thanks{CNRS, IRIF, Paris. Work done while affiliated to Universit\'e Libre de Bruxelles, Belgium, and CWI, Amsterdam, the Netherlands. {\tt smgapers@gmail.com}} \and Ronald de Wolf\thanks{QuSoft, CWI and University of Amsterdam, the Netherlands. Partially supported by the Dutch Research Council (NWO) through Gravitation-grant Quantum Software Consortium 024.003.037, and through QuantERA project QuantAlgo 680-91-034. {\tt rdewolf@cwi.nl}}
}
 
\pagenumbering{Alph}
\begin{titlepage}
\clearpage\maketitle
\thispagestyle{empty}
 
\abstract{
Graph sparsification underlies a large number of algorithms, ranging from approximation algorithms for cut problems to solvers for linear systems in the graph Laplacian.
In its strongest form, ``spectral sparsification'' reduces the number of edges to near-linear in the number of nodes, while approximately preserving the cut and spectral structure of the graph.
In this work we demonstrate a polynomial quantum speedup for spectral sparsification and many of its applications.
In particular, we give a quantum algorithm that, given a weighted graph with $n$ nodes and $m$ edges, outputs a classical description of an $\epsilon$-spectral sparsifier in sublinear time $\widetilde O(\sqrt{mn}/\epsilon)$. This contrasts with the optimal classical complexity~$\widetilde O(m)$.
We also prove that our quantum algorithm is optimal up to polylog-factors.
The algorithm builds on a string of existing results on sparsification, graph spanners, quantum algorithms for shortest paths, and efficient constructions for $k$-wise independent random strings.
Our algorithm implies a quantum speedup for solving Laplacian systems and for approximating a range of cut problems such as min cut and sparsest cut.
}
 
\end{titlepage}
\pagenumbering{arabic}

\section{Introduction and Summary}
The complexity of many graph problems naturally scales with the number of edges in the graph.
Graph sparsification aims to reduce this number of edges, while preserving certain quantities of interest.
When considering for instance the approximation of cut problems such as \mincut{} or \spcut{}, the aim is to sparsify the graph while approximately preserving its cut values.
This was first shown to be possible in the pioneering work of Karger~\cite{karger1994using} and later Bencz\'ur and Karger~\cite{benczur1996approximating}.
They introduced the concept of \emph{cut sparsifiers}, which are reweighted subgraphs that $\epsilon$-approximate all cuts in the graph.
We can then solve cut problems in the hopefully sparser subgraph, yielding an approximate solution to the original problem.
Quite surprisingly, they showed that for any undirected graph with $n$ nodes and $m$ edges, there always exists a cut sparsifier with as few as $\tO(n/\epsilon^2)$ edges, and moreover this sparsifier can be constructed in time $\tO(m)$.
This result lies at the basis of $\tO(m)$-time approximation algorithms for amongst others \mincut{}~\cite{karger1994using}, \minstcut{}~\cite{sherman2013nearly,kelner2014almost,peng2016approximate}, \spcut{} and \balsep{}~\cite{arora2009expander,sherman2013nearly}. We refer the interested reader to~\cite{picard1982selected,shmoys1997cut} for surveys on the many applications of cut approximation.
 
\begin{figure}[H]
\centering
\includegraphics[width=.95\textwidth]{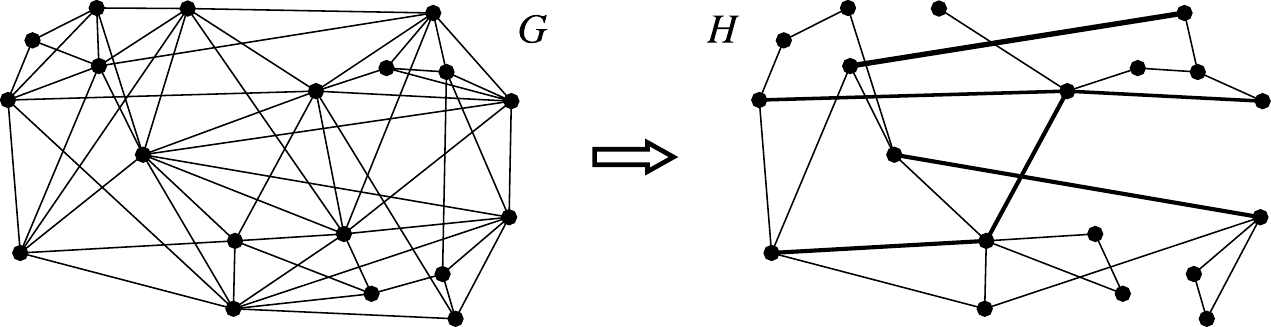}
\caption{A \emph{sparsifier} $H$ of a graph $G$ is a sparse, reweighted subgraph that preserves certain quantities such as the cut values (cut sparsifier) or quadratic forms in the graph Laplacian (spectral sparsifier).}
\label{fig:sparsifier}
\end{figure}

In their breakthrough work on Laplacian solvers, Spielman and Teng~\cite{spielman2011spectral} strengthened the notion of cut sparsifiers to so-called \emph{spectral sparsifiers}.
Rather than preserving the cut structure, these reweighted subgraphs preserve the spectral structure or \emph{quadratic form} of the Laplacian associated to the graph.
More specifically, $H$ is an $\epsilon$-spectral sparsifier of $G$ if
\[
(1-\epsilon) L_G
\preceq L_H
\preceq (1+\epsilon) L_G,
\]
with $L_H$ and $L_G$ the Laplacian matrices associated to $H$ resp.~$G$.
Since the value of any cut can be expressed as a quadratic form in the Laplacian, any spectral sparsifier is necessarily a cut sparsifier.
More importantly it implies that Laplacian systems, which are linear systems in the graph Laplacian, can be approximately solved using the Laplacian of the sparsified graph.
Similar to the case for cut sparsifiers, Spielman and Teng showed the existence and $\tO(m)$-time construction of $\epsilon$-spectral sparsifiers with $\tO(n/\epsilon^2)$ edges.
This formed a critical cornerstone of their $\tO(m)$-time solver for Laplacian systems (and faster solvers in later papers, see e.g.~\cite{jambulapati2021ultrasparse}), and the string of results and algorithms that followed it, commonly referred to as the ``Laplacian paradigm''~\cite{teng2010laplacian}.
Some examples among these are faster algorithms for learning~\cite{zhu2003semi,zhou2005learning}, computer vision and image processing~\cite{koutis2011combinatorial}, spectral clustering~\cite{vishnoi2013lx,orecchia2012approximating}, computing random walk properties~\cite{cohen2016faster}, and most recently the breakthrough almost-linear time algorithm for maximum flow and other flow problems \cite{chen2022maximum}.
The sparsification results of Spielman and Teng were later refined most notably by Spielman and Srivastava~\cite{spielman2011graph} and Batson, Spielman and Srivastava~\cite{batson2012twice}.
In \cite{batson2012twice}, the existence of spectral sparsifiers with only $O(n/\epsilon^2)$ edges was proved, which later inspired the resolution of the famous Kadison-Singer problem by Marcus, Spielman and Srivastava~\cite{marcus2015interlacing}.

\subsection{Main Result and Applications}
In this work we give a \emph{quantum} algorithm for spectral sparsification, leading to the theorem below.

\begin{theorem}[Quantum algorithm for sparsification] 
Fix $n$, $m$ and $\epsilon \geq \sqrt{n/m}$.
There exists a quantum algorithm that, given adjacency-list access to a weighted and undirected $n$-node graph $G$ with $m$ edges, outputs with high probability the explicit description of an $\epsilon$-spectral sparsifier of $G$ with $\tO(n/\epsilon^2)$ edges, in time $\tO(\sqrt{mn}/\epsilon)$.
\end{theorem}

\noindent
The algorithm outputs an explicit  classical  description, in the form of the list of $\tO(n/\epsilon^2)$ edges of the sparsifier together with their new weights.
Note the assumption $\epsilon \geq \sqrt{n/m}$. This is because sparsification is only useful when the number of edges of the sparsifier (roughly $n/\epsilon^2$) is at most the number of edges $m$ of the original graph~$G$.
Note also that $\tO(\sqrt{mn}/\epsilon) \in \tO(m)$ whenever $\epsilon \geq \sqrt{n/m}$, and hence our quantum algorithm provides a speedup over classical algorithms, whose $\tO(m)$ runtime can be shown to be optimal.\footnote{Because there is an $\Omega(m)$ query lower bound for deciding whether a graph is connected or not (see for instance \cite[Theorem~4.9, $k=1$]{eden&rosenbaum:LB} for a stronger statement), we have the same linear lower bound for finding a cut sparsifier for a given graph, as well as for applications like approximating \mincut{}.}
For dense graphs, where $m \in \Omega(n^2)$, this improves the time complexity from $\tO(n^2)$ classically to $\tO(n^{3/2})$ quantumly.

Our algorithm assumes coherent access to the input graph in the form of quantum queries to the adjacency lists. This assumption is very standard and allows us to talk about the query complexity of our algorithm. In order to also be able to talk about ``time'' complexity, we also assume a QRAM (coherent RAM) memory of $\tO(\sqrt{mn}/\epsilon)$ classical bits to which we can do classical writes, and whose bits we can query in superposition. Our algorithm uses just $O(\log n)$ ``actual'' qubits. The ``time'' (complexity) in the above theorem then measures the number of elementary gates, input queries, and QRAM writes and queries. See Section~\ref{sec:prelim} for more details about our computational model.

The algorithm builds on a range of quantum and classical results, the most important of which are classical sparsification algorithms by Spielman and Srivastava~\cite{spielman2011graph} and Koutis and Xu~\cite{koutis2014approaching}, a spanner algorithm by Thorup and Zwick~\cite{thorup2005approximate}, a quantum algorithm for single-source shortest-path trees by D\"urr, Heiligman, H{\o}yer and Mhalla~\cite{durr2006quantum} and an efficient $k$-independent hash function by Christiani, Pagh and Thorup~\cite{christiani2015independence}.

We prove a matching lower bound, showing that the runtime of our quantum algorithm is optimal up to polylog-factors.
In fact, we show that even outputting a weaker \emph{cut} sparsifier requires the same number of queries.
 
\begin{theorem}[Quantum lower bound for sparsification] \label{thm:lower-bound-intro}
Fix $n$, $m$ and $\epsilon \geq \sqrt{n/m}$.
Any quantum algorithm that, given adjacency-list access to a weighted and undirected $n$-node graph $G$ with $m$ edges, explicitly constructs with high probability an $\epsilon$-cut sparsifier of $G$ has query complexity $\widetilde{\Omega}(\sqrt{mn}/\epsilon)$.
\end{theorem}
 
Our algorithm provides a direct speedup for many of the aforementioned applications.
In Table~\ref{table:cut-appr} we illustrate this speedup for a number of cut approximation problems.
All bounds follow by combining our sparsification algorithm with the best classical algorithms, applied to the sparsifier.
As far as we know, this is the first quantum speedup for these cut approximation problems.
 
\begin{table}[H]
\centering
\def\arraystretch{1.5}%
\begin{tabular}{|c|c|c|}
\hline
& Classical & Quantum (this work) \\
\hline
.878-\maxcut & $\tO(m)$~\cite{arora2016combinatorial} & $\tO(\sqrt{mn})$ \\
\hline
$\epsilon$-\mincut & $\tO(m)$~\cite{karger2000minimum} & $\tO(\sqrt{mn}/\epsilon)$ \\
\hline
$\epsilon$-\minstcut & $\tO(m + n/\epsilon^5)$~\cite{peng2016approximate} & $\tO(\sqrt{mn}/\epsilon + n/\epsilon^5)$ \\
\hline
$O(\sqrt{\log n})$-\spcut/\textsc{bal.sep.} & $\tO(m + n^{1+\delta})$~\cite{sherman2009breaking} & $\tO(\sqrt{mn} + n^{1+\delta})$ \\
\hline
\end{tabular}
\caption{Classical and quantum time complexity of cut approximation problems. All quantum bounds follow from combining our quantum sparsification algorithm with the corresponding classical algorithm.
Parameter $\delta$ is an arbitrarily small but positive constant.} \label{table:cut-appr}
\end{table}

We can also use a classical Laplacian solver on the sparsifier to find a speedup for Laplacian solving, i.e., solving the linear system $Lx = b$ where $L$ is the Laplacian of the original graph.

\begin{theorem}[Quantum Laplacian Solver]
Fix $n$, $m$ and $\epsilon \geq \sqrt{n/m}$.
There exists a quantum algorithm that, given adjacency-list access to a weighted and undirected $n$-node graph $G$ with $m$ edges and Laplacian~$L$, outputs with high probability an approximate solution $\tilde{x}\in\mathbb{R}^n$ to the linear system $L x = b$ such that $\| \tilde{x} - x \|_L \leq \epsilon \|x\|_L$ in time $\tO(\sqrt{mn}/\epsilon)$.
\end{theorem}

\noindent
In contrast to the well-known HHL algorithm~\cite{harrow2009quantum}, our algorithm outputs an explicit  classical description of $\tilde{x}$ (i.e., a vector of $n$ real entries), not an $n$-dimensional quantum state.
Here $\|v\|_L$ denotes the $L$-induced norm $\|v\|_L = \sqrt{v^\dag L v} = \|L^{1/2} v\|$, with $v^\dag$ the complex transpose of vector $v$.
This is the typical norm considered for Laplacian solving.
This speeds up the dependency on $m$ with respect to classical solvers~\cite{spielman2014nearly}, whose runtime is $\tO(m \log(1/\epsilon))$.
Similar to the classical case, we show how this also yields a quantum speedup for solving the more general class of \textit{symmetric, weakly diagonally-dominant} (SDD) linear systems.

We also find quantum speedups for approximating effective resistances and random walk commute times, creating an approximate ``resistance oracle'' which allows to query for the effective resistance of \emph{any} node pair in time $\tO(1)$, for approximating the random walk cover time, and for approximating the bottom eigenvalues of the Laplacian.
Finally we discuss how a spectral sparsifier allows to implement spectral $k$-means clustering more efficiently, so that our quantum sparsification algorithm also leads to a speedup for this task.
We summarize our speedups in Table~\ref{table:laplacian}, and discuss prior work on quantum algorithms for some of these problems in Section~\ref{sec:prior-work}.
More details are provided in Section~\ref{sec:applications}
 
\begin{table}
\centering
\def\arraystretch{1.5}%
\begin{tabular}{|c|c|c|}
\hline
& Classical & Quantum (this work) \\
\hline
$\epsilon$-Laplacian/SDD Solving & $\tO(m)$~\cite{spielman2014nearly} & $\tO(\sqrt{mn}/\epsilon)$ \\
\hline
$\epsilon$-Effective Resistance (single) & $\tO(m)$ & $\tO(\sqrt{mn}/\epsilon)$ \\
\hline
$\epsilon$-Effective Resistances (all) & $\tO(m + n/\epsilon^4)$~\cite{spielman2011graph} & $\tO(\sqrt{mn}/\epsilon + n/\epsilon^4)$ \\
\hline
$O(1)$-Cover Time & $\tO(m)$~\cite{ding2011cover} & $\tO(\sqrt{mn})$ \\
\hline
$k$ bottom eigenvalues & $\tO(m + kn/\epsilon^2)$ & $\tO(\sqrt{mn}/\epsilon + k n/\epsilon^2)$ \\
\hline
Spectral ($k$-means) Clustering & $\tO(m + n \poly(k))$ & $\tO(\sqrt{mn} + n \poly(k))$ \\
\hline
\end{tabular}
\caption{Classical and quantum time complexity of Laplacian solving and some of its applications. All classical bounds without reference follow from \cite{spielman2014nearly}. The quantum bounds essentially follow from combining our quantum sparsification algorithm with the corresponding classical algorithm.} \label{table:laplacian}
\end{table}
 
\subsection{Quantum Algorithm}
Our quantum sparsification algorithm starts from the iterative sparsification algorithm by Koutis and Xu~\cite{koutis2016simple}.
Their algorithm provides a simple combinatorial counterpart to the usual, algebraic treatment of spectral sparsification.
It crucially relies on the growth of so-called \emph{spanners} of the graph, which are sparse subgraphs that approximately preserve all pairwise distances between nodes.
After growing a small number of disjoint spanners in the graph, and keeping these edges, they downsample the remaining edge set by keeping every edge independently with some fixed constant probability, and discarding the rest.
This results in a sparsifier with approximately half the number of edges of the original graph.
Repeating this procedure a logarithmic number of times results in an $\epsilon$-spectral sparsifier with $\tO(n/\epsilon^2)$ edges.
 
The gist of our quantum speedup comes from a faster quantum algorithm for constructing spanners.
This algorithm follows essentially by pairing a classical spanner algorithm by Thorup and Zwick~\cite{thorup2005approximate} with the shortest-paths quantum algorithm by D\"urr, Heiligman, H{\o}yer and Mhalla~\cite{durr2006quantum}.
More specifically we prove the theorem below, where we call a graph $H$ a spanner of $G$ if it is a subgraph with $O(n \log n)$ edges, and the distance between any pair of nodes in $H$ is at most $\log n$ times their original distance in $G$.
Our algorithm speeds up the classical $\tO(m)$-time algorithm by Thorup and Zwick, whose runtime is optimal.

\begin{theorem}
Fix $n$, $m$.
There exists a quantum algorithm that, given adjacency-list access to a weighted and undirected $n$-node graph $G$ with $m$ edges, outputs with high probability a spanner of $G$ in time $\tO(\sqrt{mn})$.
\end{theorem}
\noindent
We can now try to plug this faster spanner construction in the Koutis-Xu sparsification algorithm.
The problem, however, is that we cannot output the ``intermediate'' sparsifiers, since after a constant number of iterations these still have $\Omega(m)$ edges, whereas we aim for a runtime that scales with $\sqrt{mn}$.
We overcome this issue using two observations, which will allow us to describe the intermediate graphs only implicitly.
 
First we show that \emph{if} we were given query access to a uniformly random string of $\tO(m)$ bits, then we could implicitly mark the discarded edges, and grow spanners in the remaining, unmarked graph without significantly affecting the runtime.
Second, we get rid of this long random string by using that any $(k/2)$-query quantum algorithm cannot distinguish a uniformly random string from a $k$-wise independent string, which only behaves uniformly random for subsets of at most $k$ elements.
This is a known result and can be proven for instance using the polynomial method~\cite{beals2001quantum}.
Hence it suffices that we have access to a $k$-wise independent random string, allowing us to use the rich literature on \emph{$k$-independent hash functions} that aim to simulate access to such random strings.
Specifically we require the recent result by Christiani, Pagh and Thorup~\cite{christiani2015independence}, which shows that in $\tO(k)$ time we can construct a data structure that can simulate queries to a $k$-wise independent string, requiring only $\tO(1)$ time per query.
Prior to their work, all algorithms required preprocessing time $\tO(k^{1+\delta})$, for $\delta>0$.
Using their construction we can efficiently simulate the random string, which leads to the following claim.
\begin{claim}
Consider any quantum algorithm with runtime $q$ that uses a uniformly random string.
Then we can construct a quantum algorithm without random string that has the same output distribution and has a runtime $\tO(q)$.
\end{claim}
\noindent
Combining these observations remedies the issue of having to store the intermediate graphs, and leads to a speedup of the Koutis-Xu algorithm runtime to $\tO(\sqrt{mn}/\epsilon^2)$ quantum time.
 
We then further improve the runtime down to $\tO(\sqrt{mn}/\epsilon)$ by combining this quantum sparsification algorithm with the sparsification toolbox of Spielman and Srivastava~\cite{spielman2011graph}.
In that work, they show that a graph can be sparsified very elegantly by sampling edges with weights roughly proportional to their effective resistances.
Complementing this, they propose a near-linear time constructible ``resistance oracle'', which allows to query for effective resistances in logarithmic time.
We use our quantum sparsification algorithm to construct an initial, rough sparsifier with a constant error, in time $\tO(\sqrt{mn})$.
We then construct an approximate resistance oracle for this sparsifier, which effectively yields an approximate resistance oracle for the original graph.
Surprisingly, such rough approximation suffices for constructing an $\epsilon$-spectral sparsifier using the Spielman-Srivastava sampling scheme.
This finally allows us to sample the $\tO(n/\epsilon^2)$ edges of the sparsifier in time $\tO(\sqrt{mn}/\epsilon)$, using Grover's algorithm.
This idea of using a ``poor'' spectral approximation to compute sampling probabilities to obtain a better spectral approximation is also used in~\cite{li2013iterative,cohen2015uniform}.
 
\subsection{Matching Lower Bound} \label{sec:intro-lower-bnd}
We prove that the $\tO(\sqrt{mn}/\epsilon)$-runtime of our quantum algorithm is optimal, up to polylog-factors, even when we wish to construct a weaker \emph{cut} sparsifier.
The intuition behind this is that an $\epsilon$-cut sparsifier of a general graph must contain $\Omega(n/\epsilon^2)$ edges (and this is tight \cite{batson2012twice}).
If we can appropriately ``hide'' these edges among the $m$ edges of a graph, then a quantum search algorithm requires $\Theta(\sqrt{mn/\epsilon^2}) = \Theta(\sqrt{mn}/\epsilon)$ queries to retrieve them.
 
Turning this intuition into a concrete lower bound, however, turns out to be rather complicated.
We start with a random graph construction by Andoni, Chen, Krauthgamer, Qin, Woodruff and Zhang~\cite{andoni2016sketching}.
This construction describes graphs on $n$ nodes and $\tO(n/\epsilon^2)$ edges, so that any $\epsilon$-cut sparsifier must contain a constant fraction of the edges.
As such, the constructed graphs are in fact already sparsifiers.
We then carefully ``hide'' these sparsifiers in a larger, denser graph, in such a way that a sparsifier of this graph must retrieve all of the original, hidden sparsifiers.
To prove a quantum lower bound for this search problem, we describe it as the composition of the problem of finding a constant fraction of the nonzero bits in a Boolean matrix with the OR-function.
Finally we combine lower bounds for the individual problems using a composition theorem for adversary bounds, applicable to the composition of a relational problem with a function.
This composition theorem was very recently proven by Belovs and Lee \cite{belovs2019relational}, prompted by our question to them.
 
\subsection{Prior Work} \label{sec:prior-work}
We are not aware of any prior work on quantum speedups for graph sparsification.
In a very different line of work though, spectral sparsification has been studied in a quantum context with the goal of sparsifying Hamiltonian matrices, which are used to describe many-body systems.
Aharonov and Zhou~\cite{aharonov2019hamiltonian} asked whether the \emph{interaction graph} of a many-body system can be sparsified while preserving its spectrum, showing that this is not possible in general.
More recently, Herbert and Subramanian~\cite{herbert2019spectral} considered the weaker notion of sparsifying the Hamiltonian matrix, and suggested that sparsification could indeed help in Hamiltonian simulation.
They do not consider quantum algorithms for effectively constructing such a sparsifier.
 
Research on quantum algorithms for cut approximation is also limited.
There is recent work by Hamoudi, Rebentrost, Rosmanis and Santha~\cite{hamoudi2019quantum} on quantum approximate minimization of submodular functions, which can be used for cut approximation.
However, their work was more recently superseded by better classical algorithms~\cite{axelrod2020near}.
Other recent work by Brand{\~a}o, Kueng and Stilck Fran\c{c}a~\cite{brandao2019faster} used quantum SDP-solvers to approximate quadratic binary optimization problems, of which \maxcut{} is the most notable instance.
They do not succeed in finding a speedup for \maxcut{} though, mainly because their algorithm does not benefit from the special structure of this instance.
 
Concerning our speedup for Laplacian solving, we mention a range of papers on quantum speedups for general linear system solving.
Most famous is the work by Harrow, Hassidim and Lloyd~\cite{harrow2009quantum}, which was later refined in work by Ambainis~\cite{ambainis2012variable} and Childs, Kothari and Somma~\cite{childs2017quantum}.
They describe a quantum algorithm for solving general linear systems $Ax = b$ in time $\tO(d_M \kappa \log(1/\epsilon))$, with $d_M$ the row sparsity and $\kappa$ the condition number of $A$.
These algorithms are particularly relevant for sparse and well-conditioned systems (in general, however, $\kappa$ can be as large as $O(n^3 w_{\max}/w_{\min})$ for graph Laplacians \cite[Lemma 6.1]{spielman2014nearly}).
Crucially, they only output a quantum state that encodes the solution, rather than an explicit description as we do.
 
Quantum speedups for the problems of estimating effective resistances and spectral gaps have also been studied in other work.
Very recently and independently from our work, Piddock \cite{piddock2019quantum} constructed a quantum walk algorithm that $\epsilon$-approximates the effective resistance $R_{s,t}$ using $\tO(\sqrt{mR_{s,t}}/\epsilon^2) \in \tO(\sqrt{mn}/\epsilon^2)$ quantum walk steps.
He also argued how this could possibly be further improved to $\tO(\sqrt{mR_{s,t}}/\epsilon) \in \tO(\sqrt{mn}/\epsilon)$.
While the quantum walk model is different from our model, this tentative bound would agree with our runtime.
In addition, however, we can effectively approximate \emph{all} effective resistances simultaneously in the graph in the same complexity.
The problem was also studied in slightly different settings in \cite{wang2017efficient,chakraborty2018power,ito2019approximate}.
A quantum walk algorithm for estimating the second bottom eigenvalue $\lambda_2$ of the Laplacian in the adjacency-matrix model was studied by Jarret, Jeffery, Kimmel and Piedrafita~\cite{jarret2018quantum}.
They give a multiplicative $\epsilon$-approximation of $\lambda_2$ in time $\tO(n/(\sqrt{\lambda_2}\epsilon))$, which is $\tO(n^2/\epsilon)$ in the worst case.
We improve the worst-case complexity to $\tO(\sqrt{mn}/\epsilon) \in \tO(n^{3/2}/\epsilon)$.
 
We also mention some past and concurrent work on quantum speedups for clustering.
One paper by Daskin~\cite{daskin2017quantum} describes a quantum algorithm for spectral clustering but no direct speedup is found with respect to classical algorithms.
Concurrent to our work is a paper by Kerenidis and Landman \cite{kerenidis2020quantum} which describes a quantum algorithm for outputting the centroids of a $k$-means spectral clustering.
In contrast to our work, they start from quantum access to a data set, which they then use to query an associated Laplacian.
They find a quantum speedup under certain assumptions (e.g., that the input data is appropriately clustered), and as such is incomparable to our quantum algorithm.
Less directly related, there exists a number of papers~\cite{aimeur2007quantum,lloyd2013quantum,wiebe2015quantum,kerenidis2019q} on quantum speedups for $k$-means clustering and the construction of a neighborhood graph.
These tasks are complementary to our work on finding a spectral embedding, given a similarity graph of the data.
It does seem interesting to try and use these algorithms to further speed up our spectral clustering algorithm.
 
Finally, we mention some classical work on sublinear algorithms for Laplacian solving and spectral sparsification.
First, the work by Andoni, Krauthgamer and Pogrow~\cite{andoni2018solving} describes a sublinear algorithm for Laplacian solving, with the aim of approximating a single coordinate of the output.
Their algorithm is inspired by quantum algorithms for linear system solving, and similarly only finds a speedup for sparse and well-conditioned systems.
The second work is by Lee~\cite{lee2013probabilistic}, who proposes a classical algorithm for spectral sparsification of unweighted graphs which is sublinear in $m$.
He succeeds in bypassing the $\Omega(m)$ lower bound on classical sparsification by only achieving a weaker, additive error in the approximation.
As such this work is incomparable to ours.

\subsection{Follow-up work}
After the initial appearance of our results, a number of works have appeared that build on our results.

Apers and Lee \cite{apers2020edgeconn} used our quantum algorithm for graph sparsification to design a quantum algorithm for finding an \emph{exact} minimum cut in a graph.
They obtain quantum speedups when the ratio of the maximum weight over the minimum weight is bounded, and show that the quantum algorithm is optimal under such a condition.
Using similar techniques and under similar constraints, Apers, Auza and Lee \cite{apers2021sublinear} proved a quantum speedup for finding the exact minimum $s$-$t$ cut of a graph.

In a rather different setting, van Apeldoorn, Gribling, Li, Nieuwboer, Walter and de Wolf \cite{vanapeldoorn2021scaling} recently proved a quantum speedup for approximate matrix scaling.
Their algorithm builds on quantum approximate counting, and returns an $\epsilon$-approximate matrix scaling of an $n \times n$ matrix with $m$ nonzero entries in time $\tO(\sqrt{mn}/\epsilon^4)$.
For entrywise positive matrices, the upper bound was later improved to $\tO(n^{3/2}/\epsilon^3)$ by Gribling and Nieuwboer \cite{gribling2021improved}.
Their algorithm combines a second-order method for matrix scaling \cite{cohen2017matrix} with our quantum algorithm for approximately solving Laplacian systems.

In another work, Cade, Labib and Niesen \cite{cade2021quantum} described a quantum algorithm for \emph{motif clustering}, which is a graph clustering method that is based on higher-order patterns or motifs.
For dense instances, their algorithm builds on our quantum sparsification algorithm.

Finally, a recent work by Chen and de Wolf \cite{chen2021quantum} describes quantum algorithms and quantum lower bounds for linear regression with norm constraints.
The lower bounds rely on a similar strategy as ours in that they use the composition property of the adversary bound, which was recently proved by Lee and Belovs \cite{belovs2019relational}.
 
\subsection{Open Questions}
Our work raises a number of interesting questions and future directions, some of which we summarize below.
 
\begin{itemize}
\item
We prove a matching $\widetilde{\Omega}(\sqrt{mn}/\epsilon)$ lower bound on the quantum query complexity of spectral sparsification.
Can we extend this to a tight lower bound for any of the resulting applications, e.g., for $\epsilon$-approximating the min cut or effective resistance?
Since these problems can be reduced to constructing an $\epsilon$-spectral sparsifier, this would yield a stronger lower bound.
\item
The runtime of our quantum algorithm is tight, up to polylogarithmic factors, for sparsification of weighted graphs.
Can we potentially improve the runtime for \emph{unweighted} graphs?
\item
Graph sparsification is a key technique in efficient algorithms for calculating minimum cuts \cite{karger1994using,karger2000minimum}, and the use of fast Laplacian solvers played a key in the recent resolution of the long-standing question of computing maximum flows in graphs in almost-linear time~\cite{christiano2011electrical,sherman2013nearly,kelner2014almost,peng2016approximate,chen2022maximum}.
In follow-up work to this one, Apers and Lee \cite{apers2020edgeconn} used our quantum algorithm for graph sparsification to obtain a quantum speedup for finding a minimum cut in a graph.
We leave it as an open question whether a quantum speedup for finding a maximum flow can be obtained as well.
\item
Spectral sparsification of graphs and Laplacians has been extended in different directions such as sparsification of hypergraphs \cite{soma2019spectral,bansal2019new}, sparsification of sums of positive semi-definite matrices~\cite{soma2019spectral,silva2016sparse}, sparsification in a streaming setting~\cite{kelner2013spectral,kapralov2017single}.
It is also closely related to concepts such as spectral sketching~\cite{andoni2016sketching} and linear data regression using leverage scores~\cite{drineas2006sampling}.
It seems likely that we can also find quantum speedups for these related problems.
Similarly we might hope to solve more ``quantum'' tasks, such as sparsifying density operators or POVMs.
\end{itemize}

\subsection{Roadmap} \label{sec:roadmap}
To round up the introduction, we give a roadmap of the main parts of our paper.
Section \ref{sec:prelim} introduces some necessary preliminaries on graphs, the computational model, and the quantum algorithmic routines that we use.
It also describes the classical sparsification algorithm that our quantum algorithm is based on.
In Section \ref{sec:quantum-sparsification} we describe a first ``rough'' quantum sparsification algorithm, and we show how to simulate its access to a random string using $k$-wise independent hash functions.
We improve the error dependency of the algorithm in Section \ref{sec:refined-quantum-sparsification}.
In Section \ref{sec:quantum-spanner} we describe a quantum algorithm for constructing graph spanners, which our sparsification algorithms are based on.
In Section \ref{sec:lower-bound} we prove a lower bound that matches the performance of our quantum algorithm, and finally in Section \ref{sec:applications} we elaborate on the applications of our quantum algorithm for sparsification.

\section{Preliminaries} \label{sec:prelim}
 
Throughout the paper we say that something holds ``with high probability'' if it holds with probability at least $1 - O(1/n)$.
 
\subsection{Computational Model and Quantum Algorithms} \label{sec:comp-model}

%The quantum \emph{query} complexity of an algorithm is the number of queries it makes to the input graph (see Section~\ref{sec:graphs}).
%The quantum \emph{time} complexity of an algorithm is defined with respect to a computational model.
We assume as our computational model a quantum-accessible classical control system that
\begin{enumerate}
\item
can run quantum subroutines on at most $O(\log N)$ qubits, where $N$ is the size of the problem instance;
\item
can make quantum queries to the input; and
\item
has access to a quantum-read/classical-write RAM (QRAM)\footnote{Another name for this type of memory is ``coherent RAM.'' We feel a QRAM containing a classical $k$-bit string~$z$ and allowing efficient queries of the form $\ket{i,b}\mapsto\ket{i,b\oplus z_i}$ (where $i\in[k]$ and $b\in\{0,1\}$) is a reasonable generalization of classical RAM: if one believes in classical RAM and in quantum superposition, then QRAM is quite natural. Like a classical RAM, the physical hardware of such a QRAM necessarily requires size at least  proportional to~$k$ because it contains $k$ bits of information, but answering a query would have a cost proportional to $\log k$, even when querying multiple stored bits in superposition. This could be realized for instance by laying out the $k$ bits as the leaves of a binary tree of depth $\log k$; the $\log k$ bits of the binary representation of an address $i\in[k]$ would chart a path from the root to the addressed bit~$z_i$, allowing for efficient lookup of the addressed bit. Note that running this on a superposition of different different addresses $i$ would involve going down different paths in superposition, but would still only use a superposition of $O(\log k)$ qubits. 

Assuming such QRAM is very common in quantum algorithms for graph problems. It should be noted, though, that the notion of QRAM is a bit controversial, (1) because the term is sometimes used in the literature for a different and stronger kind of memory (allowing for efficient conversion of a classically-stored unit vector of $k$ numbers into the corresponding state of $\log k$ qubits), and (2) since implementing it on noisy hardware might require $O(k)$ work to error-correct a quantum query, rather than $O(\log k)$ work.} of $\tO(\sqrt{mn}/\epsilon)$ classical bits, where a single QRAM operation corresponds to either classically writing a bit to the QRAM, or making a quantum query (a read operation) to bits stored in QRAM, possibly in superposition.
\end{enumerate}
In this model, an algorithm has \emph{time complexity $T$} if it uses at most $T$ elementary classical and quantum gates, quantum queries to the input, and QRAM operations.
The \emph{query} complexity of an algorithm only measures the number of queries to the input.
If we only care about query complexity, the assumption of having QRAM may be dropped at the expense of a polynomial increase in the number of gates.

An important quantum subroutine in our work is Grover's algorithm~\cite{grover1996fast} for searching sets of marked elements, which is summarized in the claim below.
\begin{claim}[Repeated Grover Search] \label{claim:rep-grover}
Let $f:[N] \to \{0,1\}$ be a function that marks a set of elements $S = \{i \in [N] \mid f(i) = 1\}$.
Then there is a quantum algorithm that finds $S$ with probability at least $2/3$ in $\tO(\sqrt{N \, |S|})$ elementary operations and queries to $f$, and uses $O(\log N)$ qubits and a QRAM of $\tO(|S|)$ bits.
\end{claim}
\noindent
We also use a quantum algorithm for finding shortest-path trees by D\"urr, Heiligman, H{\o}yer and Mhalla~\cite{durr2006quantum}, but the quantum routines in this algorithm can be reduced to Grover search.
 
\subsection{Graphs, Queries and Spanners} \label{sec:graphs}
We consider weighted and undirected graphs $G = (V,E,w)$ with $|V|=n$ nodes and $|E|=m$ edges, and edge weights $w: E \to \real_{\geq 0}$.
We are given \emph{adjacency-list access} to $G$, as is considered in e.g.~\cite{durr2006quantum,goldreich2002property}.
This allows to query for the degree of a node, its $k$-th neighbor (according to some unknown but fixed ordering), or the weight of an edge.
This model is more restrictive than both the ``general graph model'', which in addition allows for adjacency matrix queries~\cite{goldreich2010introduction}, and the ``sparse-access model'', in which the neighbors are ordered lexicographically, as is commonly assumed in quantum algorithms for linear system solving and Hamiltonian simulation~\cite{harrow2009quantum,childs2017quantum}.
 
We define the \emph{distance} $\delta_G(u,v)$ between nodes $u$ and $v$ with respect to $G$ as
\[
\delta_G(u,v)
= \min_{u-v \;\mathrm{path}\; P} \sum_{e \in P} \frac{1}{w_e}.
\]
This definition is in accordance with the interpretation of $G$ as an electrical network, in which an edge $e$ corresponds to a link of conductance $w_e$ (and hence resistance or ``cost'' $1/w_e$), as is common in the literature on spectral sparsification.
A spanner of $G$ is a sparse subgraph $H$ that approximately preserves all pairwise distances.
Specifically, we will call $H$ a \emph{$t$-spanner} of $G$ if for any pair $u,v \in V$ it holds that
\[
\delta_G(u,v)
\leq \delta_H(u,v)
\leq t \delta_G(u,v).
\]
Note that the first inequality is trivially satisfied since $H$ is a subgraph.
It is well known that every weighted graph has a $(2k-1)$-spanner with $O(n^{1+1/k})$ edges \cite{althofer1993sparse}.
Throughout the paper we will use the shorthand \emph{spanner} to denote a $t$-spanner with $t = 2 \log n$ and $\tO(n)$ edges.
An \emph{$r$-packing of spanners} of $G$ is an ordered set $H = (H_1,H_2,\dots,H_r)$ of $r$ edge-disjoint spanners such that $H_j$ is a spanner for $G - \cup_{i<j} H_i$, which is the remaining graph after removal of the edges of all previous spanners.
Note that such an $r$-packing always exists for every~$r$, though once $G - \cup_{i<j} H_i$ has no edges left anymore, the subsequent spanners $H_j,H_{j+1},\ldots,H_r$ in the packing will all be empty.
 
The \emph{Laplacian} $L$ of a weighted graph $G$ is given by $L = D - A$, with $A$ the weighted adjacency matrix $(A_{ij}) = w_{ij}$ and $D$ the diagonal weighted degree matrix $(D_{ii}) = \sum_j w_{ij}$.
Alternatively, we can rewrite the Laplacian as
\[
L
= \sum_{e\in E} w_e \chi_e \chi_e^T,
\]
where we let $\chi_e = \chi_u - \chi_v$ denote a vector associated to the edge $e = (u,v)$, with $\chi_u,\chi_v$ indicator vectors of the nodes $u,v$ (we fix an arbitrary orientation of the edges).
If $G$ is connected then $L_G$ has a trivial kernel consisting only of the all-ones vector.
Moreover, $L_G$ is a real, symmetric, diagonally dominant matrix with nonnegative diagonal entries, and is hence positive semi-definite.
 
\subsection{Spectral Sparsification using Spanner Packings}
A \emph{cut sparsifier} $H$ of a graph $G$ is a sparse, reweighted subgraph that preserves the value of all cuts.
Specifically, $H$ is called an $\epsilon$-cut sparsifier if for any $S \subseteq V$ it holds that
\begin{equation} \label{eq:cut-spars}
(1-\epsilon) \val_G(S)
\leq \val_H(S)
\leq (1+\epsilon) \val_G(S),
\end{equation}
where $\val_G(S) = \sum_{i\in S,j\notin S} w_{(i,j)}$ denotes the total weight of the edges leaving $S$.
 
A \emph{spectral sparsifier} $H$ of a graph $G$ is a sparse, reweighted subgraph that preserves the quadratic form $x^T L_G x$ associated to the Laplacian $L_G$ of $G$, for any vector $x \in \complex^n$.
Specifically, $H$ is called an $\epsilon$-spectral sparsifier if for any $x \in \complex^n$ it holds that
\begin{equation} \label{eq:sp-spars}
(1-\epsilon) x^T L_G x
\leq x^T L_H x
\leq (1+\epsilon) x^T L_G x.
\end{equation}
Alternatively, we can rewrite this as $(1-\epsilon) L_G \preceq L_H \preceq (1+\epsilon) L_G$, where $A \preceq B$ denotes that $B - A$ is positive semi-definite.
This condition implies for instance that all eigenvalues of $H$ $\epsilon$-approximate the eigenvalues of $G$~\cite{batson2013spectral}, and all cuts in $H$ $\epsilon$-approximate those in $G$.
To see the latter, consider a subset $S \subseteq V$ and let $\chi_S$ denote the indicator on $S$, then
\[
\chi_S^T L_G \chi_S
= \sum_{(u,v)=e\in E} w_e (\chi_S(u) - \chi_S(v))^2
= \val_G(S).
\]
This shows that the cut value can be described by a quadratic form in the Laplacian, and hence \eqref{eq:sp-spars} implies that $(1-\epsilon) \val_G(S) \leq \val_H(S) \leq (1+\epsilon) \val_G(S)$, for all $S \subseteq V$.
Any $\epsilon$-spectral sparsifier is therefore also an $\epsilon$-cut sparsifier.
 
Spectral sparsifiers can be constructed by using spanners to identify the ``important'' edges in the graph.
This was first noticed by Kapralov and Panigrahy~\cite{kapralov2012spectral}, and further refined by Koutis and Xu~\cite{koutis2016simple}.
We will build on the latter work, which describes a very elegant approach for constructing spectral sparsifiers from spanner packings.
Their algorithm iteratively invokes the routine described below, which creates a spectral sparsifier with approximately half the number of edges of the original graph.
 
\begin{algorithm}[H]
\caption{$H = \halfspars(G,\epsilon)$} \label{alg:koutis-xu}
\begin{algorithmic}[1]
\State
construct an $O(\log^2(n)/\epsilon^2)$-packing of spanners of $G$
\State
let $P$ be their union and set $H = P$
\For{each edge $e \notin P$}
\State
with probability $1/4$, add $e$ to $H$ with weight $4w_e$
\EndFor
\State
return $H$
\end{algorithmic}
\end{algorithm}
 
\begin{theorem}[{\cite[Theorem 3.2]{koutis2016simple}}] \label{thm:half-sparsify}
The graph $H = \halfspars(G,\epsilon)$ is, with probability at least $1-1/n^2$, an $\epsilon$-spectral sparsifier of $G$ with at most $m/2 + \tO(n/\epsilon^2)$ edges.
\end{theorem}
\noindent
Now consider a fixed $\epsilon>0$.
If we iterate $T \in O(\log(m/n))$ times the routine $\halfspars(G,\epsilon')$, with $\epsilon' \in O(\epsilon/T)$, then we retrieve with high probability an $\epsilon$-spectral sparsifier with $\tO(n/\epsilon^2)$ edges.
By \cite{andoni2016sketching} this is optimal up to log-factors.
Classically the complexity is dominated by the construction of $\tO(1/\epsilon^2)$ spanners, each of which requires time $\tO(m)$ \cite{thorup2005approximate}, giving a total time complexity $\tO(m/\epsilon^2)$.

\section{Quantum Sparsification Algorithm} \label{sec:quantum-sparsification}
In this section we describe our quantum algorithm for constructing spectral sparsifiers.
The algorithm is based on the scheme by Koutis and Xu.
We use as a black-box a quantum algorithm for constructing a spanner in time $\tO(\sqrt{mn})$, whose description we postpone to Section~\ref{sec:quantum-spanner}.
 
As we already mentioned in the introduction, we cannot simply plug this quantum spanner algorithm in the Koutis and Xu algorithm.
Indeed, after a single iteration of their algorithm this would require to output a graph with up to $m/2$ edges, which is much too costly since we aim at a runtime that scales as $\sqrt{mn}$.
We resolve this issue in two stages.
First, we assume that we have access to a random string of length $\tO(m)$.
We use this string to mark edges that have been discarded at some iteration by 0-bits, which we later use to implicitly set their weight equal to zero.
By its construction, the spanner algorithm can then construct a spanner in the remaining graph.
At the end we use Grover search to explicitly retrieve the remaining $\tO(n/\epsilon^2)$ edges, whose union forms the spectral sparsifier.
We then get rid of the random string.
To this end we use efficient $k$-independent hash functions that allow to simulate queries to a $k$-wise independent random string.
This suffices since by standard results a $k$-query quantum algorithm cannot distinguish a $2k$-wise independent strings from a uniformly random one.
 
\subsection{Using a Random String}
We first assume access to a family of independent, random strings $r_i \in \{0,1\}^m$, with indices $i \in [\log(m/n)]$, such that all bits are independent and equal to $1$ with probability $1/4$.
For different indices $i$, the strings $r_i$ will function as consecutive ``sieves'' of the edge set.
 
Algorithm~\ref{alg:qspars} describes the sparsification algorithm using such random strings.
A critical remark is that steps 4 and 5 of the algorithm are only performed implicitly, as mentioned before.
Rather than keeping an explicit list of updated edge weights, we maintain an implicit ``weight oracle''.
Only when an edge weight is queried, does this weight oracle calculate its weight by consulting the necessary random strings.
We show how to do this efficiently in the proof of Theorem~\ref{thm:random-string}.
 
\begin{algorithm}[H]
\caption{$H = \qspars(G,\epsilon)$} \label{alg:qspars}
\begin{algorithmic}[1]
\State
let $\{w'_e = w_e\}$ and $\ell = \lceil \log(m/n) \rceil$
\For{$i = 1,2,\dots,\ell$}
\State
create an $O(\log^2(n)/\epsilon^2)$-packing of spanners of $G' = (V,E,w')$, let $P_i$ denote its union
\For{each edge $e \notin P_i$} %\Comment{\emph{implicitly!}}
\State
\textbf{if}\, $r_i(e) = 1$ \,\textbf{then}\, set $w'_e = 4w'_e$ \,\textbf{else}\, set $w'_e = 0$
\EndFor
\EndFor
\State
use repeated Grover search to find $H = \{e \in E \mid w'_e > 0\}$
\end{algorithmic}
\end{algorithm}
 
\begin{theorem} \label{thm:random-string}
Given access to independent, uniformly random strings $r_i \in \{0,1\}^m$ for $i \in [\log(m/n)]$, algorithm $\qspars(G,\epsilon)$ returns with probability $1-O(\log(n)/n^2)$ an $\epsilon$-spectral sparsifier of $G$ with $\tO(n/\epsilon^2)$ edges.
There is a quantum algorithm that implements it in time $\tO(\sqrt{mn}/\epsilon^2)$ and using a QRAM of $\tO(\sqrt{mn}/\epsilon^2)$ bits.
\end{theorem}
\begin{proof}
Correctness easily follows from Theorem~\ref{thm:half-sparsify}: in every iteration we ``half-sparsify'' the remaining graph (induced by all edges of weight $w_e > 0$).
The probability that all $\log(m/n)$ iterations succeed is $1 - O(\log(n)/n^2)$.
Below we discuss how steps 4 and 5 can be implemented efficiently, so that the runtime of the for-loop is dominated by the construction of $\tO(1/\epsilon^2)$ spanners.
By Theorem~\ref{thm:quantum-spanner} this takes time $\tO(\sqrt{mn}/\epsilon^2)$.
By standard results~\cite{nielsen2002quantum}, the repeated Grover search routine in the final step takes time $\tO(\sqrt{mn}/\epsilon)$, which is the time needed to find $n/\epsilon^2$ edges among $m$ edges.
 
What remains to prove is that there exists an efficient oracle that keeps track of the weight updates in steps 4 and 5.
Consider the $i$-th iteration.
Given an edge $e$, let $k$ denote the number of spanners before this iteration in which $e$ occurs so far.
If $k=0$, return $w'_e = 4^iw_e$ if $(r_i \, r_{i-1} \dots r_1)(e) = 1$, and $w'_e = 0$ if $(r_i \, r_{i-1} \dots r_1)(e) = 0$.
If $k>0$, let $j < i$ denote the last spanner packing in which it occurs.
Now return $w'_e = 4^{i-k}w_e$ if $(r_i \, r_{i-1} \dots r_{j+1})(e) = 1$, and $w'_e = 0$ otherwise.
This takes $\tO(1)$ searches through the set of spanners (which we may assume is sorted), and at most $O(i) \in \tO(1)$ evaluations of the random oracle.
\end{proof}
 
The space complexity of the algorithm requires $O(\log n)$ qubits and a QRAM of $\tO(n/\epsilon^2)$ bits.
The number of qubits follows from the space complexity of the quantum spanner algorithm and the Grover search routine.
The classical space complexity is dominated by the output size.
 
\subsection{Using $k$-independent Hash Functions}
In order to get rid of the random strings $\{r_i\}$, we build on the following fact, which is an easy consequence of the polynomial method~\cite{beals2001quantum}.
It seems that this was first used in the proof of~\cite[Theorem 19]{buhrman2008quantum}, and is stated explicitly in for instance~\cite[Theorem 3.1]{zhandry2015secure}.
\begin{fact} \label{fact:q-wise}
The output distribution of a quantum algorithm making $q$ queries to a uniformly random string is identical to the same algorithm making $q$ queries to a $2q$-wise independent string.
\end{fact}
As a consequence, we can replace the uniformly random strings of length $m$ by a $k$-wise independent string with $k \in \tO(\sqrt{mn}/\epsilon^2)$.
Surely we also cannot explicitly construct a $k$-wise independent string of length $\tO(m)$ in time $\tO(\sqrt{mn}/\epsilon^2)$, but we can use hash functions to simulate queries to such a string.
A family of hash functions $F = \{h:[u] \to [r]\}$ for $u, r \in \natural$ is called \emph{$k$-independent} if, for any subset $S \subseteq [u]$ of size $|S| \leq k$ and a uniformly random function $h$ in the family, the image of $h$ on $S$ behaves uniformly random in $[r]^{|S|}$.
This implies that the image of a random member of $F$, which we will refer to as a \emph{$k$-independent hash function}, describes a $k$-wise independent string over $[r]^u$.
Elegant constructions of such functions have long been known, the most famous example being random degree-$k$ polynomials, as proposed by Carter and Wegman~\cite{carter1979universal}.
Crucial to our cause, however, is that we can evaluate the hash function in $\tO(1)$ time, potentially allowing $\tO(k)$ preprocessing time.
Fortunately, such a result was established very recently by Christiani, Pagh and Thorup~\cite{christiani2015independence}, who proved the theorem below.
We note that this is a purely classical construction.
\begin{theorem}[\cite{christiani2015independence}]
It is possible to construct in time $\tO(k)$ a data structure of size $\tO(k)$ that allows to simulate queries to a $k$-independent hash function in $\tO(1)$ time per query.
\end{theorem}
\noindent
With $k = 2q$ and $[r] = \{0,1\}$, we can combine this with Fact~\ref{fact:q-wise} to give the corollary below.
\begin{corollary} \label{cor:no-random-string}
Consider any quantum algorithm with runtime $q$ that makes queries to a uniformly random string.
We can simulate this algorithm with a quantum algorithm with runtime $\tO(q)$ without random string, using an additional QRAM of $\tO(q)$ bits.
\end{corollary}
 
This shows that we can efficiently simulate the random string in Algorithm~\ref{alg:qspars}, leading to at most a polylogarithmic overhead in the runtime.
The classical space complexity of the algorithm does increase from $\tO(n/\epsilon^2)$ to $\tO(\sqrt{mn}/\epsilon^2)$.
The following theorem is immediate by combining Theorem \ref{thm:random-string} with Corollary \ref{cor:no-random-string}.
\begin{theorem} \label{thm:quantum-sparsification-1}
There exists a quantum algorithm that, given adjacency-list access to a weighted and undirected graph $G$, constructs with high probability an $\epsilon$-spectral sparsifier of $G$ with $\tO(n/\epsilon^2)$ edges in time $\tO(\sqrt{mn}/\epsilon^2)$.
The algorithm uses $O(\log n)$ qubits and a QRAM of $\tO(\sqrt{mn}/\epsilon^2)$ bits.
\end{theorem}

\section{Refined Quantum Sparsification Algorithm} \label{sec:refined-quantum-sparsification}
 
In the last section we proposed a quantum algorithm for constructing an $\epsilon$-spectral sparsifier in time $\tO(\sqrt{mn}/\epsilon^2)$.
Here we show how to improve the runtime of this algorithm to $\tO(\sqrt{mn}/\epsilon)$, which we will later show is optimal up to polylog-factors.
The improvement essentially follows from combining our previous algorithm with the seminal results on spectral sparsification by Spielman and Srivastava~\cite{spielman2011graph}.
In that work, they first showed that sampling edges with probabilities approximately proportional to their effective resistances results in a spectral sparsifier (the Koutis-Xu algorithm is derived from their result).
Then they showed how Laplacian solvers could be used to efficiently estimate these effective resistances.
We will use our quantum sparsification algorithm to first construct a ``rough'' $\epsilon$-sparsifier, for some constant $\epsilon$, which we only use to approximate the effective resistances in the original graph.
Surprisingly such approximation suffices to implement the Spielman-Srivastava sampling scheme on the original graph.
We then use a quantum sampling routine to efficiently implement this sampling scheme, finally leading to an $\epsilon$-spectral sparsifier for arbitrary $\epsilon>0$ in time $\tO(\sqrt{mn}/\epsilon)$.
This idea of using a ``poor'' spectral sparsifier for computing sampling probabilities to obtain a better spectral sparsifier is also present in for instance~\cite{li2013iterative,cohen2015uniform}.
 
\subsection{Spielman-Srivastava Toolbox and Quantum Sampling}
Here we formally introduce the main tools that we use.
These are an efficiently constructible ``resistance oracle'' and a sparsification algorithm based on this oracle from~\cite{spielman2011graph}, and a quantum sampling routine for implementing this sparsification algorithm.
 
\subsubsection{Approximate Resistance Oracle} \label{sec:resistance-oracle}
The \emph{effective resistance} in a graph $G$ between a pair of nodes $s$ and $t$ is defined as the effective resistance between $s$ and $t$ after replacing every edge $e$ by a resistor of value $1/w_e$.
It can be expressed algebraically as $R_{s,t} = (\chi_s - \chi_t)^T L_G^+ (\chi_s - \chi_t)$, where $L_G^+$ is the pseudoinverse of $L_G$ (i.e., the inverse on its image).
A Laplacian solver hence allows to efficiently compute $R_{s,t}$.
Spielman and Srivastava proved that in some sense one can efficiently compute \emph{all} effective resistances in roughly the same time.
More specifically, they showed that it is possible to construct in near-linear time a data structure of size $\tO(n/\epsilon^2)$ that allows to efficiently approximate $R_{s,t}$ for any $s,t$.
\begin{theorem}[\cite{spielman2011graph}] \label{thm:resistance-oracle}
Consider a weighted and undirected graph $G$.
There is an $\tO(m/\epsilon^2)$-time algorithm which computes a $(24 \log(n)/\epsilon^2) \times n$ matrix $Z$ such that with probability at least $1-1/n$, for every pair $s,t \in V$, it holds that
\[
(1-\epsilon) R_{s,t}
\leq \| Z (\chi_s - \chi_t) \|^2
\leq (1+\epsilon) R_{s,t}.
\]
\end{theorem}
\noindent
Hence the matrix $Z$ represents a data structure which allows to $\epsilon$-approximate $R_{s,t}$ for any pair $s,t$ by calculating the 2-norm distance between two columns, each of dimension $\tO(1/\epsilon^2)$.
 
\subsubsection{Spectral Sparsification with Edge Scores}
In the same paper, Spielman and Srivastava proved the following theorem, which shows that a spectral sparsifier can be constructed by independently keeping edges with weights roughly proportional to their effective resistances.
 
\begin{theorem}[\cite{spielman2011graph}] \label{thm:res-sparsification}
Let $2R_e \geq \tR_e \geq R_e/2$ for each edge $e \in E$, and $p_e = \min(1,C w_e \tR_e \log(n)/\epsilon^2)$ for some universal constant $C$.
Then keeping every edge $e$ independently with probability $p_e$, and rescaling its weight with $1/p_e$, yields with probability at least $1-1/n$ an $\epsilon$-spectral sparsifier of $G$ with $O(n\log(n)/\epsilon^2)$ edges.
\end{theorem}
\noindent
Note that $\sum_e p_e \gg 1$ is the expected number of edges of the sparsifier.
Since $\sum_e w_e R_e = n-1$ \cite[Theorem 25]{bollobas2013modern}, this yields the claimed number of edges.
Also note that, by slightly tweaking our estimate $\tilde R_e$, we can ensure that the sampling probabilities $p_e$ have an $O(\log n)$-bit description.
We will use this later on.
 
We note that, in fact, Spielman and Srivastava describe a slightly different scheme.
They propose to draw $\tO(n/\epsilon^2)$ independent and identically distributed edge samples from the edge set, with probability proportional to their effective resistance.
It is well-known that both schemes give the same performance bound - see e.g.~\cite[Remark 1]{fung2019general}.
 
\subsubsection{Quantum Sampling}
Assuming access to an approximate resistance oracle that gives approximations $\tR_e$ to $R_e$, we wish to implement the Spielman-Srivastava sparsification scheme.
While classically this requires time $\tO(m + \sum_e p_e)$, we can use quantum algorithms to do so more efficiently.
 
\begin{claim} \label{claim:quantum-sampling}
Assume we have query access to a list of probabilities $\{p_e\}_{e \in E}$, each of which is described with $\tO(1)$ bits of precision.
Then there is a quantum algorithm that samples a subset $S \subseteq E$, such that $S$ contains every $e$ independently with probability $p_e$, in expected time $\tO(\sqrt{m(\sum_e p_e)})$ and using a QRAM of $\tO(\sqrt{m(\sum_e p_e)})$ bits.
\end{claim}

\begin{proof}
Let $\ell \in \tO(1)$ denote the number of bits to describe each of the probabilities $p_e$.
We can assume access to a uniformly random $\ell m \in \tO(m)$-bit string $r$, since by Corollary~\ref{cor:no-random-string} this implies that there also exists a quantum algorithm without random string.
For all $e$, we can derive from this random string a random $\ell$-bit number $q_e \in [0,1]$ such that $q_e \leq p_e$ with probability exactly $p_e$.
We can implement the mapping $\ket{e} \ket{0} \mapsto \ket{e} \ket{q_e}$ in time $\tO(1)$.

%  we can derive a function $h_r: E \times [0,1] \to \{0,1\}$ such that for each $e$ independently $h_r(e,p_e) = 1$ with probability $p_e$ and $h_r(e,p_e) = 0$ otherwise.
% To this end, assume the edges in $E$ are ordered lexicographically and let $i$ denote the index of edge $e$.
% Let $q_e \in [0,1]$ denote the number described by the random bits 
% Then we set $h_r(e,p_e) = $

We combine this with a query to the list of probabilities to implement the mapping $\ket{e}\ket{0}\ket{0} \mapsto \ket{e}\ket{p_e}\ket{q_e}$.
Now we use repeated Grover search to find all edges $e$ such that $q_e \leq p_e$ -- for every edge $e$, this set contains $e$ with probability exactly $p_e$.
By Claim \ref{claim:rep-grover} this routine takes time $\tO(\sqrt{m (\sum_e p_e)})$ in expectation, which proves the lemma.
\end{proof}
 
\subsection{Refined Quantum Sparsification}
Now we will combine the Spielman-Srivastava toolbox, the quantum sampling routine and our quantum sparsification algorithm from the last section to improve the runtime of the latter from $\tO(\sqrt{mn}/\epsilon^2)$ to $\tO(\sqrt{mn}/\epsilon)$.
 
\begin{algorithm}[H]
\caption{$H = \qspars(G,\epsilon)$}
\begin{algorithmic}[1]
\State
use quantum sparsification (Theorem~\ref{thm:quantum-sparsification-1}) to construct a $(1/100)$-spectral sparsifier $H$ of $G$
\State
create a $(1/100)$-approximate resistance oracle of $H$ using Theorem \ref{thm:res-sparsification}, yielding estimates $\{\tR_e\}$
\State
use quantum sampling (Claim \ref{claim:quantum-sampling}) to sample a subset of the edges, keeping every edge with probability $p_e = \min(1,C w_e \tR_e \log(n)/\epsilon^2)$
\end{algorithmic}
\end{algorithm}
 
\begin{theorem}[Quantum Spectral Sparsification] \label{thm:main}
$\qspars(G,\epsilon)$ returns with high probability an $\epsilon$-spectral sparsifier $H$ with $\tO(n/\epsilon^2)$ edges, and has runtime $\tO(\sqrt{mn}/\epsilon)$.
The algorithm uses $O(\log n)$ qubits and a QRAM of $\tO(\sqrt{mn}/\epsilon)$ bits.
\end{theorem}
\begin{proof}
First we prove correctness.
Since $H$ is a spectral sparsifier of $G$, and effective resistances correspond to quadratic forms in the inverse of the Laplacian, we know that the effective resistances of $H$ approximate those of $G$: $(1-1/100) R^G_{s,t} \leq R^H_{s,t} \leq (1+1/100) R^G_{s,t}$ for all $s,t \in V$.
By Theorem~\ref{thm:resistance-oracle} we know that the approximate resistance oracle yields estimates $\{\tR^H_{s,t}\}$ such that $(1-1/100) R^H_{s,t} \leq \tR^H_{s,t} \leq (1+1/100) R^H_{s,t}$.
Combining these inequalities shows that
\[
(1-1/100)^2 R^G_{s,t}
\leq \tR^H_{s,t}
\leq (1+1/100)^2 R^G_{s,t}.
\]
If we now keep each edge with probability $p_e = \min(1,C w_e \tR^H_e \log(n)/\epsilon^2)$, then by Theorem~\ref{thm:res-sparsification} with probability $1-1/n$ we find an $\epsilon$-spectral sparsifier with $O(n \log(n)/\epsilon^2)$ edges.
Combining this success probability with those of the quantum sparsification algorithm and the construction of the resistance oracle, we find a total success probability of at least $(1-1/n)^3 = 1 - O(1/n)$.
 
The bound on the runtime follows from summing the $\tO(\sqrt{mn})$ runtime of the quantum sparsification algorithm, the $\tO(n)$ runtime for creating the resistance oracle of the sparsifier with $\tO(n)$ edges, and the $\tO(\sqrt{m(\sum_e p_e)})$ expected runtime of the quantum sampling routine (as remarked after Theorem \ref{thm:res-sparsification}, we can assume that the probabilities $p_e$ have an $O(\log n)$-bit description).
Since
\[
\sum_e p_e
\leq \frac{C \log(n)}{\epsilon^2} \sum_e w_e \tR^H_e
\leq \frac{(1+1/100)^2 C \log(n)}{\epsilon^2} \sum_e w_e R^G_e,
\]
and $\sum_e w_e R^G_e = n-1$~\cite[Theorem 25]{bollobas2013modern}, we have that $\sum_e p_e \in \tO(n/\epsilon^2)$ and so the expected runtime of the sampling routine is $\tO(\sqrt{m n}/\epsilon)$.
Moreover, by the Chernoff bound the runtime of the latter routine will indeed be $\tO(\sqrt{mn}/\epsilon)$ with probability at least $1-1/n$.
Hence we can abort the algorithm whenever the runtime exceeds this bound, and the algorithm will still succeed with high probability, while the total runtime becomes $\tO(\sqrt{mn}/\epsilon)$ in the worst case.
\end{proof}

\section{Quantum Algorithm for Building Spanners} \label{sec:quantum-spanner}
 
The Koutis-Xu sparsification algorithm identifies ``important'' edges by growing spanners inside the graph.
In this section we propose a quantum algorithm for growing spanners, speeding up the best classical algorithms.
 
Recall from Section~\ref{sec:prelim} that a $t$-spanner of a graph $G = (V,E,w)$ is a subgraph $H = (V,E_H \subseteq E,w)$ that preserves all pairwise distances between nodes up to a stretch factor $t$.
For every pair $u,v \in V$, it should hold that
\[
\delta_G(u,v)
\leq \delta_H(u,v) \leq t \delta_G(u,v),
\]
where we recall that $\delta_G(u,v) = \min_{u-v \;\mathrm{path}\; P} \sum_{e \in P} 1/w_e$.
We illustrate this in Figure \ref{fig:spanner}.
A spanner preserves the original weights on its edges.
This is in contrast to spectral sparsifiers which are necessarily reweighted.
A classic result by Alth\"ofer et al.~\cite{althofer1993sparse} shows that, for any parameter $k>0$, any $n$-node graph has a $(2k-1)$-spanner with $O(n^{1+1/k})$ edges.
We refer the interested reader to the classic book by Peleg~\cite{peleg2000distributed} or the very recent survey by Ahmed et al.~\cite{ahmed2019graph}.
 
\begin{figure}
\centering
\includegraphics[width=.75\textwidth]{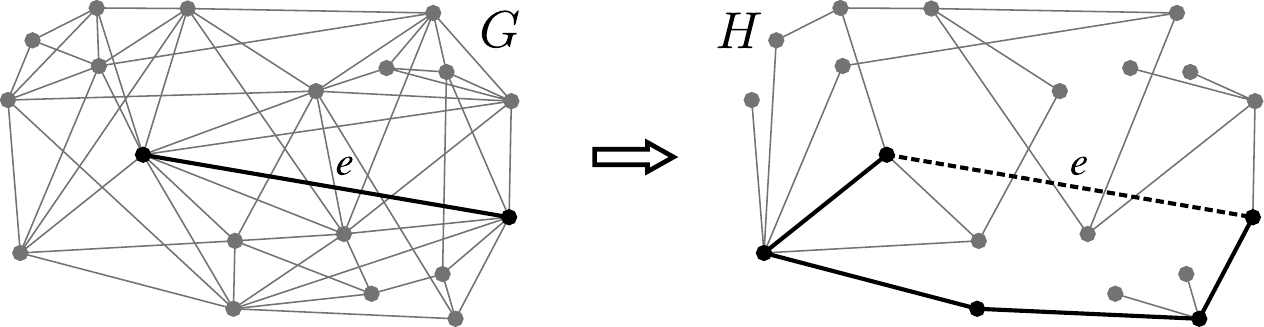}
\caption{A \emph{$t$-spanner} $H$ of a graph $G$ is a sparse subgraph that preserves all shortest path distances between pairs of vertices up to a \emph{stretch} factor $t$. Equivalently, for any edge $e = (x,y)$ in $G$ there exists a path between $x$ and $y$ of distance at most $t/w_e$ in $H$.}
\label{fig:spanner}
\end{figure}
 
There exists a range of classical algorithms for constructing spanners.
We will make use of one by Thorup and Zwick~\cite{thorup2005approximate}, which follows from their work on ``approximate distance oracles''.
The main bottleneck of their algorithm is the growth of shortest-path trees in subgraphs.
We speed up this bottleneck by using the quantum algorithm of D\"urr, Heiligman, H{\o}yer and Mhalla~\cite{durr2006quantum} for growing a shortest-path tree in time $\tO(\sqrt{mn})$.
 
\subsection{Thorup-Zwick Algorithm} \label{sec:thorup-zwick}
The spanner algorithm from~\cite{thorup2005approximate} makes use of \emph{shortest-path trees} (SPTs).
A shortest-path tree $T(v)$ from a node $v$ spanning a subset $C$ is defined as a tree, rooted at $v$ and spanning $C$, so that the distance in this tree from $v$ to any node in $C$ is the same as their distance in the original graph~$G$.
 
The Thorup-Zwick algorithm, presented in Algorithm~\ref{alg:spanner}, randomly partitions the node set into $k$ layers $\{A_i\}$, which are increasingly sparsified.
The nodes in these layers function as ``hubs'' for the nearby nodes.
Shortest-path trees are then grown that allow efficient routing along these hubs.
The resulting spanner consists of the union of these shortest-path trees.
In the algorithm below, we set $\delta(w,\emptyset) = \infty$ for any $w \in V$.
 
\begin{algorithm}[H]
\caption{$H = \spanner(G,k)$} \label{alg:spanner}
\begin{algorithmic}[1]
\State
let $A_0 = V$ and $A_k = \emptyset$
\For{$i = 1,2,\dots,k$}
\State
if $i<k$, let $A_i$ contain each element of $A_{i-1}$, independently, with probability $n^{-1/k}$
\For{$v \in A_{i-1} - A_i$}
\State
grow shortest-path tree $T(v)$ from $v$ spanning $C(v) = \{ w \in V \mid \delta(w,v) < \delta(w,A_i) \}$
\State
add $T(v)$ to $H$
\EndFor
\EndFor
\end{algorithmic}
\end{algorithm}

Figure \ref{fig:spanner-constr} illustrates the use of shortest-path trees inside the algorithm. 
Apart from the correctness of the algorithm, we will require some additional bounds on the size of the intermediate clusters $C(v)$.
We extract the following theorem from the analysis by Thorup and Zwick.
 
\begin{theorem}[\cite{thorup2005approximate}] \label{thm:thorup-zwick}
\leavevmode
\begin{itemize}
\item
The output graph $H$ of $\spanner(G,k)$ is a $(2k-1)$-spanner of $G$.
\item
The expected number of edges in $H$ is $O(\Exp(\sum_v |C(v)|)) \in O(k n^{1+1/k})$.
\item
The expected number of edges with at least one node in the clusters is
\[
\Exp(\sum_v |E(C(v))|) \in O(kmn^{1/k}).
\]
\end{itemize}
\end{theorem}
\noindent
Setting $k = 1/2 + \log n$, as we will do later on, this yields a $2\log n$-spanner with an expected number of edges $O(n \log n)$.
 
\begin{figure}
\centering
\includegraphics[width=.45\textwidth]{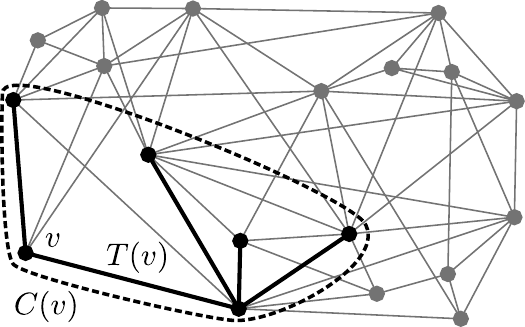}
\caption{The spanner algorithm by Thorup and Zwick \cite{thorup2005approximate}. In each iteration, the algorithm grows a shortest-path tree $T(v)$ from a vertex $v$ until it spans some subset $C(v)$.}
\label{fig:spanner-constr}
\end{figure}

\subsection{Quantum Spanner Algorithm} \label{sec:quantum-alg-spanner}
 
We can use a quantum algorithm from D\"urr, Heiligman, H{\o}yer and Mhalla~\cite{durr2006quantum} to speed up the construction of the shortest-path tree $T(v)$, spanning $C(v)$.
We slightly generalize their algorithm to deal with ``forbidden edges'', which are encoded by associating a weight $w_e = 0$ to them (which corresponds to an infinite resistance or cost).
Such edges will correspond to edges going outside of $C(v)$, as well as edges that have already been discarded by our sparsification algorithm.
 
In Appendix~\ref{app:SPT} we prove the following statement.
We define the connected component of a node $v_0$ as the smallest subset $C_{v_0} \subseteq V$ such that $v_0 \in C_{v_0}$ and either $E(C_{v_0},V \backslash C_{v_0}) = \emptyset$ or $\max\{w_e \mid e \in E(C_{v_0},V \backslash C_{v_0})\} = 0$.
This implies that there is no path of finite distance between $v_0$ and any node outside $C_{v_0}$.
\begin{proposition}
Assume adjacency-list access to a weighted and undirected graph $G = (V,E,w)$.
Let $v_0$ be a source node and $C_{v_0}$ its connected component.
Then there exists a quantum algorithm that outputs, with probability at least $1-\delta$, a shortest-path tree from $v_0$ that spans $C_{v_0}$.
It has a runtime $\tO(\sqrt{|C_{v_0}| |E(C_{v_0})|} \log(n/\delta))$ and requires $O(\log n)$ qubits and a QRAM of $\tO(|C_{v_0}|)$ bits.
\end{proposition}
 
From this we can speed up the spanner construction rather straightforwardly.
To see this, note that the runtime of the Thorup-Zwick algorithm is dominated by the task of growing the shortest-path trees $T(v)$, spanning the local clusters $C(v)$, for all nodes $v \in V$.
By setting $w_e = 0$ for any edge reaching out of $C(v)$, this task corresponds to a shortest-path tree on the connected component of $v$.
If we use the above quantum algorithm to accelerate this, the total runtime becomes
\[
\tO\bigg( \sum_v \sqrt{|C(v)| |E(C(v)|} \bigg)
\in \tO\Bigg( \sqrt{\sum_v |C(v)|} \sqrt{\sum_v |E(C(v)|} \Bigg),
\]
where the containment follows from the Cauchy-Schwarz inequality.
By Theorem~\ref{thm:thorup-zwick} we know that $\Exp(\sum_v |C(v)|) \in O(kn^{1+1/k})$ and $\Exp(\sum_v |E(C(v))|) \in O(kmn^{1/k})$.
By Markov's inequality this implies that with probability close to 1 the runtime is
\[
\tO\Big(\sqrt{k n^{1+1/k}} \sqrt{k m n^{1/k}}\Big)
\in \tO\big(k n^{1/k} \sqrt{mn}\big).
\]
What remains to be shown is how we (implicitly) set $w_e = 0$ for all edges reaching out of $C(v)$.
To that end we follow the idea of Thorup and Zwick of connecting a new source node $s$ to every node in~$A_i$, with edges of infinite weight, and construct an SPT from $s$ to $V$.
It is easy to see that this returns the shortest path from any node $w \notin A_i$ to $A_i$, allowing to calculate $\delta(w,A_i)$.
Using the standard quantum SPT algorithm of \cite{durr2006quantum} we can construct this SPT in time $\tO(\sqrt{mn})$, and we do this whenever we construct a new $A_i$.
Now assume that the quantum SPT algorithm at some point wishes to choose an edge $(w,w')$, with $w$ part of the SPT constructed so far, and $w'$ an adjacent node.
Then by design this must be a cheapest border edge of the SPT constructed so far, and $\delta(v,w') = \delta(v,w) + \delta(w,w')$.
Hence we know $\delta(v,w')$ and we can simply check whether $\delta(v,w') < \delta(w,A_i)$, setting the weight of the edge equal to zero if this is not the case.
This proves the following theorem.
\begin{theorem} \label{thm:quantum-spanner}
There is a quantum algorithm that outputs in time $\tO(k n^{1/k} \sqrt{mn})$ with high probability a $(2k-1)$-spanner of $G$ of size $O(k n^{1 + 1/k})$.
The algorithm uses $O(\log n)$ qubits and a QRAM of $\tO(k n^{1+1/k})$ bits.
\end{theorem}
\noindent
Setting $k = \log n + 1/2$, we find an $\tO(\sqrt{mn})$ quantum algorithm for constructing $2\log n$-spanners, as is required by our sparsification algorithm.

\section{Matching Lower Bound: A Hidden Sparsifier} \label{sec:lower-bound}
In this section we prove that the runtime of our quantum algorithm for spectral sparsification is optimal, up to polylog-factors.
In fact, we show that even constructing a weaker \emph{cut} sparsifier requires the same complexity.
The following is a rephrasing of Theorem \ref{thm:lower-bound-intro} from the introduction.
\begin{theorem} \label{thm:quantum-lower-bnd}
Fix $n$, $m$ and $\epsilon \geq \sqrt{n/m}$.
Consider the problem of outputting, with high probability, an explicit description of an $\epsilon$-cut sparsifier of a weighted, undirected graph $G$ with $n$ nodes and $m$ edges, given adjacency-list access to $G$.
The quantum query complexity of this problem is $\widetilde{\Omega}(\sqrt{mn}/\epsilon)$.
\end{theorem}
\noindent
Note that sparsification is only meaningful under the constraint $\epsilon \geq \sqrt{n/m}$, since for $\epsilon \in O(\sqrt{n/m})$ the sparsifier would have at least as many edges as the original graph.
We prove this lower bound by ``hiding'' a sparsifier in a larger graph, and then proving a quantum lower bound for finding the sparsifier.
More specifically, we use a random graph construction by Andoni, Chen, Krauthgamer, Qin, Woodruff and Zhang~\cite{andoni2016sketching} such that any cut sparsifier must contain a constant fraction of the edges of the graph.
We then hide a number of copies of this random graph by embedding it in a larger, denser graph.
Finally we show that finding a constant fraction of the edges of the initial random graph requires $\widetilde{\Omega}(\sqrt{mn}/\epsilon)$ queries.
To prove this lower bound, we combine a quantum lower bound for the OR-function with an information-theoretic lower bound for the problem of finding nonzero bits in a Boolean matrix.
We can combine these separate lower bounds by using a composition theorem for adversary bounds, applicable to the composition of a relational problem with a function.
Prompted by our question, such a composition theorem was very recently proven by Belovs and Lee \cite{belovs2019relational}.
 
\subsection{Hiding a Sparsifier}
 
We use a random graph construction of Andoni et al.~\cite{andoni2016sketching} for which a sparsifier must output a constant fraction of its edges.
We then carefully hide a number of copies of this graph into a larger graph, which will later allow for the reduction of a query problem to the construction of a sparsifier.
 
\subsubsection{An Unsparsifiable Graph} \label{sec:andoni}
Andoni et al.~\cite{andoni2016sketching} construct a communication problem that is described by a family of random graphs on $2/\epsilon^2$ nodes with $1/(2\epsilon^4)$ edges.
They show that for a constant fraction of the inputs ($>3/5$), the communication problem requires to communicate $\Omega(1/\epsilon^4)$ bits.
On the other hand, for at least a $2/3$-fraction of the inputs, the communication problem can be solved by communicating an $\epsilon$-cut sparsifier.
This shows that for at least a $(3/5-1/3)>1/4$-fraction of the inputs, the description of an $\epsilon$-cut sparsifier requires $\Omega(1/\epsilon^4)$ bits.
Using a slightly more refined argument, they even show that the number of edges of the sparsifier for these instances must be $\Omega(1/\epsilon^4)$.
 
Fix any $\epsilon > 0$.
Let $B_\epsilon$ be any bipartite graph with $1/\epsilon^2$ nodes on each side, and every left node connected to a corresponding subset of half of the right nodes.
From Andoni et al.~\cite[Theorem 3.3]{andoni2016sketching} we can extract the following theorem.
Indeed, if the claim would not hold, then we could violate their communication lower bound by communicating an $\epsilon$-cut sparsifier.
 
\begin{theorem}[\cite{andoni2016sketching}] \label{thm:andoni-sketch}
For at least a $1/4$-fraction of all graphs $B_\epsilon$, any $\epsilon$-cut sparsifier must contain $\Omega(1/\epsilon^4)$ edges.
\end{theorem}
\noindent
It follows that at least a $1/4$-fraction of all graphs $B_\epsilon$ cannot be significantly sparsified.
Similarly to~\cite{andoni2016sketching}, we will also consider larger families of disjoint copies of $B_\epsilon$.
We can easily prove the following corollary using the Chernoff bound.
\begin{corollary} \label{cor:andoni}
Consider the disjoint union of $\ell$ distinct copies of $B_\epsilon$.
There exists a constant $\eta>0$, independent of $\ell$, such that for at least a $9/10$-fraction of all such graphs it holds that any $\epsilon$-cut sparsifier must contain at least $\eta \ell/\epsilon^4$ edges.
\end{corollary}
 
\subsubsection{Embedding the Sparsifier} \label{sec:hiding-sparsifier}
 
Fix $n$, $m \leq n^2/4$ and $\epsilon \geq \sqrt{n/m}$.
Consider $\ell = \epsilon^2 n/2$ independent copies $B^{(k)}$ of $B_\epsilon$, yielding a graph with $n$ nodes.
We wish to ``hide'' this graph in a larger, denser graph.
To this end, we use an $m$-bit string $x$ to (redundantly) describe the resulting graph, which we denote by $B(x)$.
The description $x$ consists of $\ell = \epsilon^2 n/2$ matrices $x^{(k)}$ of dimension $1/\epsilon^2 \times 1/\epsilon^2$,
\[
x^{(k)}
= \{x^{(k)}_{i,j} \mid i,j \in [1/\epsilon^2]\}, \quad k \in [\epsilon^2 n/2].
\]
Every matrix $x^{(k)}$ is used to describe the bipartite adjacency matrix of the corresponding copy $B^{(k)}$.
Rather than bits, however, the entries $x^{(k)}_{i,j}$ correspond to smaller strings of $N = 2\epsilon^2 m/n$ bits each, with at most one nonzero bit per string.
We say that an input $x$ is \emph{valid} if it is of this form.
The presence of an edge between the $i$-th left node and the $j$-th right node in $B^{(k)}$ is determined by the presence of a nonzero bit in the string $x^{(k)}_{i,j}$, i.e., by $\ORf(x^{(k)}_{i,j})$.
The bipartite adjacency list of the $k$-th copy $B^{(k)}$ is hence described as
\[
\mathrm{adj}(B^{(k)})
= \ORf(x^{(k)})
= \ORf\left( \begin{bmatrix}
00010 & 00000 & 00000 & 10000 \\
00100 & 00000 & 00100 & 00000 \\
00000 & 01000 & 00001 & 00000 \\
00000 & 00000 & 00010 & 01000
\end{bmatrix} \right)
= \begin{bmatrix} 1 & 0 & 0 & 1 \\ 1 & 0 & 1 & 0 \\ 0 & 1 & 1 & 0 \\ 0 & 0 & 1 & 1 \end{bmatrix},
\]
where we give concrete values to the bits in $x^{(k)}$ for illustration.
In the next section we prove a lower bound on the identification of a constant fraction of the 1-bits in a valid input $x$.
In this section we show how to embed the corresponding graph $B(x)$ in a larger, denser graph $G(x)$, so that an $\epsilon$-cut sparsifier of $G(x)$ must identify a constant fraction of the edges of $B(x)$.
This identifies a constant fraction of the 1-bits of the input $x$, so that the aforementioned lower bound effectively yields a lower bound for the construction of a sparsifier.
We must take particular care in embedding $B(x)$ in $G(x)$ so as not to reveal additional information about $x$.
E.g., we must prevent that the degrees of $G(x)$ reveal anything about the location of 1-bits in the input.
We do this essentially by ensuring that a query to the adjacency list of $G(x)$ can be performed using a single query to $x$.
 
Initially, the oblivious (input-independent) mother graph $G = G_1 \cup G_2$ is the disjoint union of two bipartite graphs $G_1 = (L_1 \cup R_1,E_1)$ and $G_2 = (L_2 \cup R_2,E_2)$.
The first graph $G_1$ is a multigraph consisting of $\ell = \epsilon^2 n/2$ disjoint copies $B^{(k)}$ of the \emph{complete} bipartite graph on $1/\epsilon^2$ nodes, containing $N = 2\epsilon^2 m/n$ parallel copies of each edge (we will get rid of these multi-edges later).
As such, $G_1$ has exactly $n$ nodes and $m$ (sometimes parallel) edges.
We match every edge of $G_1$ to a unique input bit by matching the parallel copies of edge $(i,j)$ in $B^{(k)}$ to the input bits in the string $x^{(k)}_{i,j}$.
The second bipartite graph $G_2$ has at most $n$ nodes and exactly $m$ edges, but no multi-edges (we do not need to further specify this graph at this point).
We formally match every single copy of an edge in $G_1$ to a unique edge in $G_2$, so that every input bit now corresponds to a unique edge in $G_1$ and a unique edge in $G_2$.
While doing so, we ensure that all edges leaving a left or right node of $G_1$ are matched to edges in $G_2$ whose left (resp.~right) ends are distinct.
We will later clarify why this is important, and defer a proof that this is possible for some choice of $G_2$ to Appendix \ref{app:existence}.
At this point, all edges in $G$ have zero weight.
 
Next we take the input $x$ into account, turning $G$ into $G(x)$.
To this end, we ``flip'' edge pairs conditioned on the input bit.
Specifically, consider a bit $x(i)$ that corresponds to edges $(l,r)$ in $G_1$ and $(l',r')$ in $G_2$.
If $x(i) = 1$, then we keep these edges as they are, except that we give $(l,r)$ a unit weight.
If $x(i) = 0$, then we ``flip'' the edges: we replace $(l,r)$ and $(l',r')$ by edges $(l,l')$ and $(r,r')$.
This is illustrated in Figure \ref{fig:matching}.
We see now that if two outgoing edges from $l$ were matched to edges with the same left end $l'$ in $L_2$, then this could create the edge $(l,l')$ twice.
Similarly if two incoming edges from $r$ were matched to the same right ends $r'$, this could create the edge $(r,r')$ twice.
Our matching ensures that this can never happen.
 
\begin{figure}
\centering
\includegraphics[width=.6\textwidth]{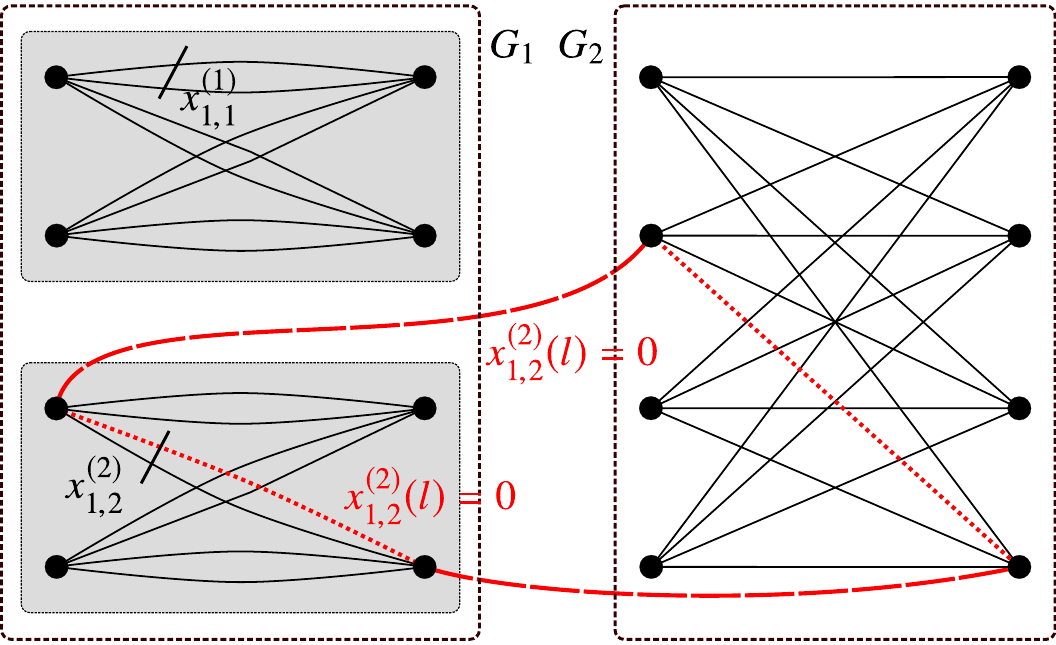}
\caption{Matching input bits to edges in $G_1$ and $G_2$ for $n = 8$, $m = 16$ and $\epsilon = 1/\sqrt{2}$ (we set $G_2$ to be the complete bipartite graph for illustration only). The dotted red edges depict a pair of matched edges, which correspond to input bit $x^{(2)}_{1,2}(l)$.
If $x^{(2)}_{1,2}(l) = 1$, these edges are kept in $G(x)$; if $x^{(2)}_{1,2}(l) = 0$, they are ``flipped'' with the dashed edges.}
\label{fig:matching}
\end{figure}
 
Now consider a pair $l \in L_1$ and $r \in R_1$, corresponding to the edge $(i,j)$ in $B^{(k)}$.
Then $G(x)$ will contain a unique, unit-weight edge between $l$ and $r$ if and only if the string $x^{(k)}_{i,j}$ has a unique nonzero bit.
As a consequence, $G(x)$ restricted to the node set $L_1 \cup R_1$ exactly describes $B(x)$.
Moreover, we can perform a single query to the adjacency list of $G(x)$ using a single query to $x$, as we prove in the lemma below.
 
\begin{lemma} \label{lem:embedded-sparsifier}
Fix $n$, $m \leq n^2/4$ and $\epsilon \geq \sqrt{n/m}$, and consider a valid input $x \in \{0,1\}^m$.
Then $G(x)$ has at most $2n$ nodes and exactly $2m$ edges, and any query to its adjacency list can be simulated using a single query to $x$.
$G(x)$ has $B(x)$ as a subgraph with unit edge weights, and all remaining edges have zero weight.
\end{lemma}
\begin{proof}
We only prove the claim about the query access, as all other claims follow easily from the construction.
A degree query is trivially simulated, since $G(x)$ has the same degrees as $G$.
To simulate a neighbor query, say that we wish to query the $k$-th entry of the adjacency list of node $l \in L_1$ in $G_1$, for some $k \in [2m/n]$.
By construction, this corresponds to a known edge $(l,r)$ in $G_1$, a corresponding edge $(l',r')$ in $G_2$ and a unique input bit $x(i)$.
We query the input bit $x(i)$.
If it equals 1, then $G(x)$ has the unit-weight edge $(l,r)$ and hence we return $r \in R_1$, and weight 1.
If it equals 0, then we flipped the edges and $G(x)$ has the edge $(l,l')$, and so we return the node $l' \in L_1$ and weight 0.
The same reasoning applies to queries from all other nodes of $G(x)$.
\end{proof}
 
We will also use the claim below, which follows easily from the fact that only the subgraph $B(x)$ of $G(x)$ has unit-weight edges, and all remaining edges have weight zero.
\begin{claim}
Any $\epsilon$-cut sparsifier of $G(x)$ contains an $\epsilon$-cut sparsifier of $B(x)$.
\end{claim}
\begin{proof}
Let $H$ be any $\epsilon$-cut sparsifier of $G(x)$.
Then we can simply remove all edges from $H$ that are not in $B(x)$ (and necessarily have zero weight) to yield an $\epsilon$-cut sparsifier for $B(x)$.
\end{proof}
 
Combining this claim with Corollary~\ref{cor:andoni}, we can deduce that for at least $9/10$-th of all valid inputs~$x$, any $\epsilon$-cut sparsifier of $G(x)$ must identify a constant fraction (specifically, at least $\eta n/\epsilon^2$) of the edges of $B(x)$.
Since the presence of an edge in $B(x)$, say $(l_i,r_j)$ in copy $B^{(k)}$, reveals that string $x^{(k)}_{i,j}$ has a nonzero entry, we find the following corollary.
 
\begin{corollary} \label{cor:reduction}
For at least a $9/10$-fraction of all valid inputs $x$, it holds that any $\epsilon$-cut sparsifier of $G(x)$ must identify a constant $2\eta$-fraction of the nonzero strings.
\end{corollary}

\subsection{Finding the Sparsifier}
 
Now we wish to prove a lower bound on identifying nonzero strings, thereby proving a lower bound on the complexity of sparsifying $G(x)$.
To this end, we first formalize the new problem.
 
\begin{definition}
The problem $\findb_{r,c}$ takes as input an $r \times c$ Boolean matrix, with each row containing exactly $c/2$ nonzero bits.
A correct output consists of the indices of a $2\eta$-fraction of the nonzero bits.
\end{definition}
 
\noindent
Note that this is a \emph{relational} problem: for every input there are many different possible correct outputs. 
Formally, a relational problem $f$ corresponds to a set $f\subseteq S\times T$, where $S\subseteq \{0,1\}^M$ is the set of possible inputs and $T$ denotes the set of possible outputs. Output $t$ is deemed a correct output for $f$ on input $x\in S$, if and only if $(x,t)\in f$.
For our relational problem $\findb_{r,c}$, $S$ is the allowed set of $r \times c$ Boolean input matrices, and each $t\in T$ corresponds to a set of $\eta rc$ indices of that matrix. Output $t$ is correct for input matrix $x$, if all indices in $t$ correspond to 1-bits in~$x$. 

Using information theory, we can prove a lower bound on the bounded-error quantum query complexity of this problem, i.e., on the number of queries to~$x$ required by a quantum algorithm that returns a correct output with probability at least $2/3$.

\begin{claim} \label{claim:bnd-findb}
The bounded-error quantum query complexity of $\findb_{r,c}$ on any constant fraction of its valid inputs is $\widetilde{\Omega}(rc)$.
\end{claim}
\begin{proof}
This lower bound follows from combining an information-theoretic lower bound on how much information the algorithm extracts about its input, with Holevo's quantum information bound.
More precisely, we examine the mutual information $I(A;B)$ between a uniformly random input matrix~$A$, and an output $B$ for $\findb_{r,c}$ that is correct with probability at least $2/3$.
We refer the reader to the textbook of Cover and Thomas \cite{cover2012elements} for precise definitions.
By Holevo's theorem~\cite{holevo1973bounds}, this mutual information lower bounds the quantum query complexity times $O(\log(rc))$.\footnote{This is a fairly standard argument. We can consider a communication protocol where the $T$-query quantum algorithm implements each quantum query using one round of $O(\log(rc))$ qubits of communication to and from another party that holds $A$ (by sending the query and answer registers back and forth), thus learning $I(A;B)$ bits of information about $A$ while only communicating $O(T\log(rc))$ qubits. It now follows from \cite{holevo1973bounds} that $T \in \Omega(I(A;B)/\log(rc))$.}
 
We lower bound $I(A;B)=H(A)-H(A|B)$ by lower bounding $H(A)$ and upper bounding $H(A|B)$.
To lower bound $H(A)$, let $S$ denote the set of all inputs, consisting of all $r \times c$ Boolean matrices with all row sums equal to $c/2$, and let $S' \subseteq S$ denote any constant fraction subset of $S$ (say $|S'| = \beta |S|$ for some $\beta>0$).
Then $A$ corresponds to a uniformly random element of $S'$, and we have
\begin{align*}
H(A)
= \log |S'|
&= \log (\beta |S|) \\
&= \log\bigg( \beta \prod_{j=1}^r \binom{c}{c/2} \bigg) \\
&= r \log \binom{c}{c/2} + \log\beta
\geq r(c - O(\log c)).
\end{align*}
To upper bound $H(A|B)$, we use that with probability at least $2/3$ the output $B$ is correct, revealing a $2\eta$-fraction of the nonzero entries of $A$ and hence significantly reducing the entropy of $A$.
Let $E$ denote an indicator bit which equals 1 if $B$ is a correct output, and 0 otherwise.
Then we can bound
\begin{align*}
H(A|B)
&\leq H(A|B,E) + H(E)\\
&= \Pr(E=1) \, H(A|B,E=1) + \Pr(E=0) \, H(A|B,E=0) + H(E) \\
&\leq \frac{2}{3} H(A|B,E=1) + \frac{1}{3}rc + 1.
\end{align*}
We now want to upper bound $H(A|B,E=1)$.
Fix integers $\{d(j)\}_{j\in[r]}$ such that $\sum_j d(j) = \eta rc$, and condition on the event $D_{\{d(j)\}}$ that $B$ reveals $d(j)$ entries of the $j$-th row of $A$.
Let $S_{\{d(j)\}} \subseteq S'$ denote the set of inputs which are compatible with $\{d(j)\}$.
We can bound
\begin{align*}
H(A|B,E=1,D_{\{d(j)\}})
& \leq \log|S_{\{d(j)\}}|
\leq \log\bigg( \prod_{j=1}^r \binom{c-d(j)}{c/2-d(j)} \bigg)\\
& = \sum_{j=1}^r \log \binom{c-d(j)}{c/2-d(j)}
\leq \sum_{j=1}^r (c - d(j))
= r(1-\eta)c.
\end{align*}
This upper bound holds irrespective of the specific setting of $\{d(j)\}$.
Hence, taking the expectation over all possible settings of $\{d(j)\}$, we obtain
\[
H(A|B,E=1) = \Exp_{\{d(j)\}} \left[ H(A|B,E=1,D_{\{d(j)\}}) \right] \leq r(1-\eta)c.
\]
Our lower bound on $H(A)$ and upper bound on $H(A|B)$ together imply
\[
I(A;B)
\geq \frac{2}{3}r(\eta c - O(\log c)) - 1
\in \Omega(rc),
\]
which finalizes our proof.
\end{proof}
 
Next we will compose the relational problem $\findb_{r,c}$ with $r \times c$ copies of the $\ORf_N$-function on $N$ bits each, so that every single input bit of $\findb_{r,c}$ is now described by an $\ORf_N$.
We denote this composed problem as $\findb_{r,c} \circ \ORf^{rc}_N$.
To ensure that the input of the composed problem is a valid input for $\findb_{r,c}$, we must restrict it to $r \times c$ matrices of strings, each carrying $N$ bits, so that exactly $c/2$ strings per row have one nonzero entry, and the remaining strings only have zeros.
 
The composed problem $\findb_{r,c} \circ \ORf^{rc}_N$ corresponds to the composition of a relational problem (multiple outputs are correct for $\findb_{r,c}$) and a function.
Bounds on the quantum query complexity of the composition of two functions are well understood \cite{hoyer2007negative}.
When the outer problem is a relational problem, however, no such result seemed to be known.
Fortunately, prompted by our question, Belovs and Lee~\cite{belovs2019relational} very recently proved the following theorem.

\begin{theorem}[Corollary~27 of~\cite{belovs2019relational}] \label{thm:adv-comp}
Let $f\subseteq S\times T$, with $S\subseteq\{0,1\}^M$, be a relational problem with bounded-error quantum query complexity~$L$.
Assume that $f$ is efficiently verifiable, in the sense that there is a bounded-error quantum algorithm that, given some $t \in T$ and oracle access to $x \in S$,
verifies whether $(x,t) \in f$ using $o(L)$ queries to~$x$.
Let $g:D\to\{0,1\}$, with $D\subseteq\{0,1\}^N$, be a Boolean function whose bounded-error quantum query complexity is $Q$.
Then the bounded-error quantum query complexity of the relational problem $f \circ g^M$, restricted to inputs $x\in\{0,1\}^{NM}$ such that $g^M(x) \in S$, is $\Theta(LQ)$.
\end{theorem}
 
We instantiate this by setting $f$ to the relational problem $\findb_{r,c}$, with $S$ a $9/10$-fraction of its valid inputs.
By Claim~\ref{claim:bnd-findb} its bounded-error quantum query complexity is  $L\in\widetilde{\Omega}(rc)$.
Using Grover's algorithm, we can efficiently verify an output, using $O(\sqrt{rc})\in o(L)$ queries to $f$'s $M$-bit input.
We let $g$ be the function $\ORf_N$ restricted to the set $D$ of all $N$-bit inputs of Hamming weight 0 or 1; its bounded-error quantum query complexity is $\Theta(\sqrt{N})$.
Plugging this into Theorem~\ref{thm:adv-comp} yields the following corollary.

\begin{corollary} \label{cor:lower-bound-bit}
Solving the problem $\findb_{r,c} \circ \ORf^{rc}_N$ has bounded-error quantum query complexity $\widetilde{\Omega}(rc\sqrt{N})$.
This holds even when the inputs are restricted to a constant fraction of the valid inputs.
\end{corollary}
 
Finally, using the graph embedding of the previous section, we show in the claim below how to solve this composed relational problem by constructing a cut sparsifier of an associated graph.
Combining this with the lower bound from Corollary~\ref{cor:lower-bound-bit} then yields our lower bound for finding a cut sparsifier (Theorem~\ref{thm:quantum-lower-bnd}).
 
\begin{claim}
Fix $n$, $m$, $\epsilon \geq \sqrt{n/m}$, and set $r = n/2$, $c = 1/\epsilon^2$ (number of potential neighbors) and $N = 2\epsilon^2 m/n$.
For at least a $9/10$-fraction of all valid inputs $x$, we can reduce the problem $(\findb_{r,c} \circ \ORf^{rc}_N)(x)$ to finding an $\epsilon$-cut sparsifier of $G(n,m,\epsilon,x)$.
Specifically, any $T$-query quantum algorithm that sparsifies $G(n,m,\epsilon,x)$ with success probability at least $2/3$ can solve $(\findb_{r,c} \circ \ORf^{rc}_N)(x)$ with success probability at least $2/3$ on at least a $9/10$-fraction of its valid inputs using $O(T)$ queries.
\end{claim}
\begin{proof}
Consider a valid input $x \in \{0,1\}^m$ and the associated graph $G(x) = G(n,m,\epsilon,x)$ as defined in Section \ref{sec:hiding-sparsifier}.
For a $9/10$-fraction of these inputs, we know by Corollary \ref{cor:reduction} that we can solve $(\findb_{r,c} \circ \ORf^{rc}_N)(x)$ by constructing an $\epsilon$-cut sparsifier of $G(x)$.
Moreover, by Lemma \ref{lem:embedded-sparsifier} we can query the adjacency list of $G(x)$ using a single query to $x$.
Hence, if we output with probability at least $2/3$ an $\epsilon$-cut sparsifier of $G(x)$ using $T$ queries to its adjacency list, then we also solve $(\findb_{r,c} \circ \ORf^{rc}_N)(x)$ with probability at least $2/3$ using $T$ queries to $x$.
\end{proof}

\section{Applications} \label{sec:applications}
 
In this section we non-exhaustively list some applications of our quantum sparsification algorithm, proving speedups in cut approximation and Laplacian solving.
Given the broad applicability of sparsification in classical algorithms, we expect that more applications will follow.
 
\subsection{Cut Approximation}
A range of cut approximation algorithms have a near-linear runtime in the number of edges in the graph.
In the following we discuss a number of these, and show how our quantum algorithm for cut sparsification allows to speed them up.
 
\subsubsection{\maxcut{}}
The \maxcut{} problem for a weighted graph $G = (V,E,w)$ asks for a cut $(S,S^c)$ that maximizes the total weight $\val_G(S)$ of the cut.
Its decision version is one of the 21 problems famously shown to be NP-complete by Karp~\cite{karp1972reducibility}.
The best current approximation factor of this problem is roughly $.8785$, as was famously proven by Goemans and Williamson~\cite{goemans1995improved} using an SDP relaxation.
Khot et al.~\cite{khot2007optimal} showed that this approximation factor is optimal under the unique games conjecture.
 
In later work, Arora and Kale~\cite{arora2016combinatorial} showed how to solve the Goemans-Williamson SDP in time $\tO(m)$, using amongst others the cut sparsification algorithm by Bencz\'ur and Karger~\cite{benczur1996approximating}.
We find a quantum speedup by replacing the Bencz\'ur-Karger algorithm by our quantum sparsification algorithm.
Specifically we construct an $\epsilon$-spectral sparsifier $H$, for some small but constant $\epsilon>0$, in time $\tO(\sqrt{mn}/\epsilon)$.
Since all cuts are preserved up to a multiplicative error $\epsilon$, we can apply the Arora-Kale algorithm on $H$ to retrieve a \maxcut{} approximation factor of at least $.8785 (1-\epsilon)$.
Choosing $\epsilon$ sufficiently small we find the following claim.
 
\begin{claim}
There exists a quantum algorithm that outputs a $.878$-approximate max cut of a weighted graph in time $\tO(\sqrt{mn})$.
\end{claim}
 
We note that for \emph{unweighted} graphs, \maxcut{} can be approximately solved classically in time $\tO(n)$.
If we wish to output the \maxcut{} bipartition, then this is trivially optimal both for classical and quantum algorithms, and hence no quantum speedup is possible.
The $\tO(n)$ classical algorithm follows from the fact that for unweighted graphs a trivial sparsification procedure suffices to approximately preserve the \maxcut{} value (pick $\tO(n)$ edges uniformly at random---see for instance~\cite[Section 2]{trevisan2012max}).
In the adjacency list model, this can be done classically in time $\tO(n)$.
For weighted graphs this approach no longer works, as in this case the edges have to be sampled with probability proportional to their edge weights.
This cannot be done classically in time $o(m)$ since we could for instance hide a single heavy edge among $m$ light edges.
From all the cut problems we consider, this is the only one for which a classical sublinear algorithm with multiplicative error exists - albeit only for the unweighted case.
 
\subsubsection{\mincut{}}
Given a weighted graph $G = (V,E,w)$, the \mincut{} problem asks for a cut $(S,S^c)$ with minimum weight $\val_G(S)$.
Up to polylog-factors in the runtime, the current best algorithm is by Karger~\cite{karger2000minimum}.
It builds on cut sparsification to output an exact min cut of the graph with high probability in time $\tO(m)$.
Running this algorithm on our $\epsilon$-cut sparsifier with $\tO(n/\epsilon^2)$ edges thus requires time $\tO(n/\epsilon^2)$, and returns an $\epsilon$-approximate min cut.
 
\begin{claim}
There exists a quantum algorithm that outputs an $\epsilon$-approximate min cut of a weighted graph in time $\tO(\sqrt{mn}/\epsilon)$.
\end{claim}
\noindent
We leave it as an open question whether this algorithm can be improved to also output an \emph{exact} min cut, similar to the algorithm in \cite{karger2000minimum}.
 
\subsubsection{\minstcut{}}
Given a weighted graph $G = (V,E,w)$, the \minstcut{} problem requires to output a cut $C = (S,S^c)$ with minimum value $\val_G(S)$, such that $s \in S$ and $t \notin S$.
The current best algorithms for approximate \minstcut{} build on the max-flow min-cut theorem, which states that the value of the min $st$-cut equals the value of a maximum $st$-flow.
Combining classical cut sparsification with the recent $\tO(m/\epsilon^3)$ solver for $\epsilon$-approximate max flows by Peng~\cite{peng2016approximate}, this yields an $\tO(m + n/\epsilon^5)$ time algorithm for \minstcut{}.
We obtain the following claim by running this algorithm on our cut sparsifier with $\tO(n/\epsilon^2)$ edges.
 
\begin{claim}
There exists a quantum algorithm that outputs an $\epsilon$-approximate minimum $st$-cut of a weighted graph in time $\tO(\sqrt{mn}/\epsilon + n/\epsilon^5)$.
\end{claim}
 
\subsubsection{\spcut{} and \balsep{}}
Given a weighted graph $G = (V,E,w)$, the \spcut{} problem asks for a cut $C = (S,S^c)$ which minimizes the ratio $\val_G(S)/(|S|\,|S^c|)$.
The \balsep{} problem asks in addition that the cut is ``balanced'', i.e., $\mu \leq |S|/|V| \leq 1/2$ for some constant $\mu>0$.
Exactly solving either of these problems is NP-hard, and optimally trying to approximate them has led to an interesting line of research~\cite{leighton1999multicommodity,arora2009expander}.
Currently the best polynomial-time classical algorithm for both problems achieves an $O(\sqrt{\log n})$-approximation factor in time $\tO(m + n^{1+\delta})$, for an arbitrary positive constant $\delta$.
This is achieved by combining cut sparsification, work by Sherman~\cite{sherman2009breaking} which shows that an $O(\sqrt{\log n})$-approximation can be calculated by solving $\tO(n^\delta)$ max flows, and the $\tO(m/\epsilon^3)$ max flow algorithm by Peng~\cite{peng2016approximate} for constant $\epsilon>0$.
Applying these algorithms to our $\epsilon$-cut sparsifier for constant $\epsilon>0$, we get the claim below.
 
\begin{claim}
There is a quantum algorithm that outputs an $O(\sqrt{\log n})$-approximate sparsest cut or balanced separator of a weighted graph in time $\tO(\sqrt{mn} + n^{1+\delta})$, for an arbitrary constant $\delta > 0$.
\end{claim}
 
\subsection{Quantum Laplacian Solver}
The complexity of the best known linear system solver is $\tO(n^\omega)$, with $\omega < 2.373$ the matrix multiplication coefficient.
Building on a long line of work, Spielman and Teng \cite{spielman2004nearly} famously showed that the special case of Laplacian systems can be solved in time $\tO(m) \in \tO(n^2)$, with $m$ the number of nonzero entries of the Laplacian.
To this end they exploit the connection of Laplacians with graphs, allowing for a combinatorial interpretation and treatment of the linear system.
 
More specifically, a Laplacian solver aims to solve a linear system $L_G x = b$, where $L_G$ is the Laplacian associated to some graph $G$.
We denote the solution by $x = L_G^+ b$.
The original motivation for spectral sparsification was in fact to create better Laplacian solvers.
Indeed, as follows from the lemma below, we can solve the Laplacian system in the sparsifier and retrieve an approximation of the original system.
Here we use the $A$-induced norm $\|v\|_A = \sqrt{v^\dag A v} = \|A^{1/2} v\|$ for a positive semi-definite matrix $A$, with $v^\dag$ the Hermitian transpose of vector $v$.
 
\begin{claim} \label{claim:approx-Laplacian}
Consider a linear system $L_G x = b$, where $L_G$ is the Laplacian of a weighted, undirected graph $G$.
If $H$ is an $\epsilon$-spectral sparsifier of $G$, with Laplacian $L_H$, then solving $L_H x = b$ yields an approximate solution to the original system:
\[
\| L_H^+ b - x \|_{L_G}
\leq 2\epsilon \|x\|_{L_G}.
\]
\end{claim}
\begin{proof}
Since $H$ is an $\epsilon$-spectral sparsifier of $G$, we have that $(1-\epsilon) L_G \preceq L_H \preceq (1+\epsilon) L_G$.
This implies that the nonzero eigenvalues of $L_G L_H^+$ (and hence also of $L_G^{1/2} L_H^+ L_G^{1/2}$) lie between $1/(1+\epsilon)$ and $1/(1-\epsilon)$.
With $I$ the identity matrix restricted to the image of $L_G$ and $L_H$, this implies that $\| L_G^{1/2} L_H^+ L_G^{1/2} - I \| \leq \epsilon/(1-\epsilon) \leq 2\epsilon$ for $\epsilon \leq 1/2$.
From this we can bound
\begin{align*}
\| L_H^+ b - x \|_{L_G}
&= \| L_G^{1/2} L_H^+ b - L_G^{1/2} x \| \\
&= \| L_G^{1/2} L_H^+ L_G x - L_G^{1/2} x \| \\
&= \| (L_G^{1/2} L_H^+ L_G^{1/2} - I) L_G^{1/2} x \|
\leq \| L_G^{1/2} L_H^+ L_G^{1/2} - I \| \; \| L_G^{1/2} x \|.
\end{align*}
Since $\| L_G^{1/2} L_H^+ L_G^{1/2} - I \| \leq 2\epsilon$ and $\| L_G^{1/2} x \| = \|x\|_{L_G}$, this proves the lemma.
\end{proof}
 
This observation allowed Spielman and Teng to reduce the task of solving a potentially dense Laplacian system to solving a very sparse one in the Laplacian of the sparsifier, having $\tO(n/\epsilon^2)$ edges.
They then invoke other methods to efficiently solve the sparse Laplacian system in additional time $\tO(n/\epsilon^2)$.
 
An immediate consequence is that we can use our quantum sparsification algorithm to speed up Laplacian solving. We can create a sparsifier $H$ with $\tO(n/\epsilon^2)$ edges in time $\tO(\sqrt{mn}/\epsilon)$, and then use a classical Laplacian solver to solve the system $L_H x = b$ in an additional time $\tO(n/\epsilon^2)$.
This proves the proposition below.
 
\begin{proposition}[Quantum Laplacian Solver]
There exists a quantum algorithm that, given adjacency-list access to a weighted and undirected graph $G$, outputs with high probability an approximate solution $\tilde{x}$ to the linear system $L_G x = b$ such that $\| \tilde{x} - x \|_{L_G} \leq \epsilon \|x\|_{L_G}$ in time $\tO(\sqrt{mn}/\epsilon)$.
\end{proposition}
 
The use of the $L_G$-induced norm $\|\cdot\|_{L_G}$ is common in the study of Laplacian solvers.
If, however, an $\epsilon$-approximate solution in the regular 2-norm is desired, then one can use that $\| \tilde{x} - x \|_{L_G} \leq \epsilon \|x\|_{L_G}$ implies that $\|\tilde{x}-x\| \leq \epsilon \sqrt{\kappa} \|x\|$, with $\kappa = \lambda_{\max}/\lambda_{\min}$ the condition number of $L_G$.
Setting $\epsilon = \delta/\sqrt{\kappa}$ hence yields a $\delta$-approximation in the 2-norm, and the runtime of quantum Laplacian solver becomes $\tO(\sqrt{mn \kappa}/\delta)$.
This factor-$\sqrt{\kappa}$ overhead would typically be too large, apart from the case where the Laplacian is well-conditioned.
A resolution to this issue follows from improving the error-dependence of our quantum Laplacian solver down to polylogarithmic, which is one of the open questions that we mentioned in the introduction.
 
\subsubsection{Solving SDD Systems}
Laplacian systems are special cases of the more general class of \textit{real, symmetric, weakly-diagonally dominant} (SDD) systems $Ax = b$, where $A$ is such that
\[
A = A^T \qquad\text{ and }\qquad A_{ii} \geq \sum_{j \neq i} |A_{ij}|, \quad \forall i \in [n].
\]
A Laplacian system $Lx = b$ is the special case of an SDD system where the off-diagonals of $L$ are nonpositive, and its row sums are equal to zero.
It is well known \cite{gremban1996combinatorial} that classically solving SDD systems can be reduced in time $O(m)$ to solving a Laplacian system.
This implies that the near-linear time Laplacian solver in \cite{spielman2014nearly} yields a near-linear time solver for the more general case of SDD systems as well.
 
Since we are aiming for a sublinear runtime in the quantum case, this classical reduction is too costly.
However, taking inspiration from the classical reduction, we show in Appendix \ref{app:SDD} that our quantum sparsification algorithm can approximately reduce an SDD system to a much sparser SDD system, in time $\tO(\sqrt{mn}/\epsilon)$.
Using a classical near-linear time SDD solver \cite{spielman2014nearly} on this sparser SDD system then yields the following proposition.
As is customary in quantum linear algebra routines \cite{gilyen2019quantum}, we assume \textit{sparse access} to the matrix $A$: given a query $(i,k) \in [n]^2$, this returns the index and value $(j,A_{ij}) \in [n] \times \real$ of the $k$-th nonzero element of the $i$-th row $A_{i\cdot}$ (and an error symbol if there are less than $k$ nonzero elements).
 
\begin{proposition}[Quantum SDD Solver] \label{prop:quantum-SDD}
There exists a quantum algorithm that, given sparse access to an SDD matrix $A$, outputs with high probability an approximate solution $\tilde{x}$ to the linear system $A x = b$ such that $\| \tilde{x} - x \|_A \leq \epsilon \|x\|_A$ in time $\tO(\sqrt{mn}/\epsilon)$.
\end{proposition}
 
\subsubsection{Effective Resistances and Commute Times}
Electrical networks, consisting of nodes $\{v \in V\}$ and resistors $\{r_e \,|\, e \in E\}$, are conveniently described by a weighted graph $G = (V,E,\{w_e = 1/r_e\})$~\cite{bollobas2013modern}.
Certain quantities of the electrical network can then be expressed using the Laplacian of $G$.
One example is the effective resistance $R_{s,t}$ between a pair of nodes $s$ and $t$, which we already encountered in Section~\ref{sec:refined-quantum-sparsification}.
This can be expressed as a quadratic form in the inverse of the Laplacian:
\[
R_{s,t}
= (\chi_s - \chi_t)^T L^+ (\chi_s - \chi_t).
\]
This effectively measures the \emph{dissipated power} $\mathcal{E}(j_{s,t})$ of the electric flow that results from injecting a unit current in $s$, and extracting it from $t$, as is described by the demand vector $j_{s,t} = \chi_s - \chi_t$.
For a more general demand vector $j \in \real^n$, with $1^T j = \sum_{v \in V} j(v) = 0$, the dissipated power is defined analogously:
\[
\mathcal{E}(j)
= j^T L^+ j.
\]
Closely related to the effective resistance is the \emph{commute time} of a random walk.
The random walk commute time $C_{s,t}$ between nodes $s$ and $t$ is defined as the expected number of steps the walk must take from $s$ to reach $t$, and then return to $s$.
By a result of Chandra, Raghavan, Ruzzo, Smolensky and Tiwari~\cite{chandra1996electrical} we know that $C_{s,t} = 2WR_{s,t}$, with $W = \sum_e w_e$ the total edge weight in the graph.
 
Since the effective resistance and the dissipated power correspond to quadratic forms in the inverse of the Laplacian, they can be $\epsilon$-approximated by calculating the corresponding quantity in an $\epsilon$-spectral sparsifier.
In addition, the total edge weight $W_H$ of an $\epsilon$-cut sparsifier $\epsilon$-approximates the original edge weight $W$, so that $(1-\epsilon)^2 C_{s,t}^G \leq C_{s,t}^H \leq (1+\epsilon)^2 C_{s,t}^G$ and hence also the commute times are approximated.
Using our quantum sparsification algorithm, together with a classical Laplacian solver, we can approximate any of these quantities in time $\tO(\sqrt{mn}/\epsilon)$.
\begin{claim}
Let $j \in \real^n$ be a current demand vector, with $1^T j = 0$.
Then there exists a quantum algorithm that outputs an $\epsilon$-approximation to the dissipated power $\mathcal{E}(j)$ in time $\tO(\sqrt{mn}/\epsilon)$.
In particular, if $j = \chi_s - \chi_t$ then this yields an $\epsilon$-approximation of the effective resistance $R_{s,t}$ and the commute time $C_{s,t}$.
\end{claim}
 
By a similar argument, we can create an $\epsilon$-spectral sparsifier, and construct the approximate resistance oracle of Spielman and Srivastava (see Section~\ref{sec:resistance-oracle}) on this sparsifier.
By Theorem~\ref{thm:resistance-oracle} we can construct this oracle in time $\tO(n/\epsilon^4)$ (which corresponds to $\tO(m/\epsilon^2)$ when we input a sparsifier with $m \in \tO(n/\epsilon^2)$ edges).
This proves the following claim.
\begin{claim}
There exists an $\tO(\sqrt{mn}/\epsilon + n/\epsilon^4)$-time quantum algorithm that outputs a $(24\log(n)/\epsilon^2) \times n$ matrix $Z$ such that with high probability, for any $s,t \in V$ it holds that
\[
(1-\epsilon) R_{s,t}
\leq \| Z (\chi_s - \chi_t) \|^2
\leq (1+\epsilon) R_{s,t}.
\]
\end{claim}
\noindent
After creating the matrix $Z$, we can hence $\epsilon$-approximate \emph{any} effective resistance $R_{s,t}$ in time $\tO(1/\epsilon^2)$, simply by calculating the norm of the difference between two $\tO(1/\epsilon^2)$-length columns of $Z$.
Apart from its use for spectral sparsification, as we demonstrated in Section~\ref{sec:refined-quantum-sparsification}, such an oracle can also be used to obtain geometric embeddings of the graph, see for instance~\cite{von2007tutorial}.
 
\subsubsection{Cover Time}
The cover time $\tcov(G)$ of a weighted graph $G$ denotes the expected number of steps before a random walk has visited all nodes, starting from the worst initial node.
A classic bound on the cover time called \emph{Matthew's bound}~\cite{matthews1988covering} states that
\[
\max_{s,t} H_{s,t}
\leq \tcov(G)
\leq (1 + \log(n)) \max_{s,t} H_{s,t},
\]
with $H_{s,t}$ the hitting time from node $s$ to node $t$.
Tighter characterizations were later proven by Kahn, Kim, Lov\'asz and Vu~\cite{kahn2000cover} and Ding, Lee and Peres~\cite{ding2011cover}.
We extract the following claims from the latter.
\begin{theorem}{{\cite[Theorems 1.6 and 4.14]{ding2011cover}}}
\leavevmode
\begin{itemize}
\item
If $H$ is an $O(1)$-spectral sparsifier of $G$, then $\tcov(H) \in \Theta(\tcov(G))$.
\item
There is an algorithm that outputs in time $\tO(m)$ with high probability an $O(1)$-approximation of $\tcov(G)$.
\end{itemize}
\end{theorem}
\noindent
We can easily derive the following claim by using our quantum sparsification algorithm to construct an $O(1)$-spectral sparsifier, and then approximating the cover time on this sparsifier using the algorithm from~\cite{ding2011cover}.
\begin{claim}
Let $G$ be an unweighted, undirected graph.
There exists a quantum algorithm that outputs a constant-factor approximation to the cover time $\tcov(G)$ in time $\tO(\sqrt{mn})$.
\end{claim}
 
\subsubsection{Eigenvalues and Spectral Clustering}
The bottom eigenvalues and eigenvectors of a graph Laplacian provide useful information about the graph, as is witnessed by e.g.~Cheeger's inequality and spectral clustering~\cite{ng2002spectral}, the PageRank algorithm~\cite{brin1998anatomy}, and in fact the entire field of spectral graph theory~\cite{cvetkovic1980spectra,chung1997spectral}.
Laplacian solvers allow to efficiently approximate these eigenvalues and eigenvectors.
Spielman and Teng~\cite{spielman2014nearly} (with a later refinement by Koutis, Levin and Peng~\cite{koutis2016faster}) showed that in time $\tO(m + kn/\epsilon^2)$ it is possible to compute (i) an $\epsilon$-approximation to the $k$ smallest eigenvalues $\lambda_1,\dots,\lambda_k$ of a Laplacian, and (ii) a set of $k$ orthogonal unit vectors $v_1,\dots,v_k$ such that
\begin{equation} \label{eq:appr-eigvectors}
v_\ell^T L v_\ell
\leq (1+\epsilon) \lambda_\ell,\quad 1 \leq \ell \leq k.
\end{equation}
This set of vectors approximates the subspace spanned by the $k$ bottom eigenvectors of the Laplacians.
Already for constant $\epsilon>0$, such a set can be used for spectral clustering.
For the case of two clusters, this is explicitly discussed in~\cite{spielman2007spectral} and~\cite[Section 7]{spielman2014nearly} for RatioCut.
For $k$ clusters, the most common approach is spectral $k$-means clustering~\cite{ng2002spectral,von2007tutorial}.
Using the same analysis as in~\cite{peng2015partitioning}, one can show that a set obeying \eqref{eq:appr-eigvectors} can be used to obtain the same performance~\cite{zanetti2019}.
 
Using our quantum sparsification algorithm, we find a direct speedup for this task.
To that end, note that it suffices to calculate the $k$ smallest eigenvalues and approximate eigenvectors of an $\epsilon'$-spectral sparsifier, for say $\epsilon' = \epsilon/10$.
This will yield $\epsilon$-approximate eigenvalues and eigenvectors of the original graph - see e.g.~\cite[Proposition 7.3]{spielman2014nearly}.
We can hence construct an $(\epsilon/10)$-spectral sparsifier with $\tO(n/\epsilon^2)$ edges in time $\tO(\sqrt{mn}/\epsilon)$, and then use a classical algorithm to solve the problem in the sparsifier, taking additional time $\tO(k n/\epsilon^2)$.
This proves the following claim.
\begin{claim}
There exists an $\tO(\sqrt{mn}/\epsilon + kn/\epsilon^2)$-time quantum algorithm that outputs with high probability an $\epsilon$-approximation of each of the $k$ smallest eigenvalues and a set of orthogonal unit vectors $v_1,\dots,v_k$ such that $v_\ell^T L v_\ell \leq (1+\epsilon) \lambda_\ell$ for all $1 \leq \ell \leq k$.
\end{claim}
 
This directly yields a speedup for the aforementioned spectral clustering algorithms, provided that we are given adjacency-list access to some similarity graph of the data \cite{von2007tutorial}.
Consider for instance the spectral $k$-means clustering algorithm in~\cite{ng2002spectral}.
Given the similarity graph, its classical time complexity is dominated by (i) the time to construct a set of $k$ vectors obeying \eqref{eq:appr-eigvectors} for constant $\epsilon>0$, which is $\tO(m + nk)$, and (ii) the time to perform $k$-means clustering on these vectors, which is $\tO(n \poly(k))$.
This yields a total classical complexity $\tO(m + n \poly(k))$.
From the above claim, we immediately find the following corollary.
\begin{corollary}
There exists a quantum algorithm that, given adjacency-list access to the similarity graph of a data set, performs spectral $k$-means clustering on this data set in time $\tO(\sqrt{mn} + n \poly(k))$.
\end{corollary}
 
\section*{Acknowledgements}
We are very grateful to Aleksandrs Belovs and Troy Lee for proving a composition property of the adversary method for relational problems which we required for the proof of our lower bound, and to Joran van Apeldoorn and Tijn de Vos for spotting an error in the quantum shortest-path tree algorithm from \cite{durr2006quantum}, which we have fixed here.
This work also benefited from discussions with Andr\'e Chailloux, Shantanav Chakraborty, Andr\'as Gily\'en, Michael Kapralov, Robin Kothari, Anthony Leverrier, Christian Majenz, Christian Schaffner, Luca Trevisan and Luca Zanetti.
We also thank the anonymous FOCS referees for helpful comments.

%%%%%%%%%%%%%%%%%%%%%%%%%%%%%%%%%%%%%%%%%%%%%%%%%%%%%%%%%%%%%%%%%%%%%%%%%%%%%%%%%%%%%%%
\bibliographystyle{alpha}
\bibliography{biblio}
%%%%%%%%%%%%%%%%%%%%%%%%%%%%%%%%%%%%%%%%%%%%%%%%%%%%%%%%%%%%%%%%%%%%%%%%%%%%%%%%%%%%%%%

\appendix

%%%%%%%%%%%%%%%%%%%%%%%%%%%%%%%%%%%%%%%%%%%%%%%%%%%%%%%%%%%%%%%%%%%%%
\section{Quantum Algorithm for Shortest-Path Trees} \label{app:SPT}
%%%%%%%%%%%%%%%%%%%%%%%%%%%%%%%%%%%%%%%%%%%%%%%%%%%%%%%%%%%%%%%%%%%%%
In this section we describe a quantum algorithm for constructing a single-source shortest-path tree (SPT), based\footnote{It was brought to our attention by J.~van Apeldoorn and T.~de Vos that there is an error in the original algorithm from \cite{durr2006quantum}. Namely, in the proof of \cite[Theorem 7.1]{durr2006quantum} the cost function $f((u,v)) = 1/w(u,v)$ is used, while this should have been $f((u,v)) = \delta_G(v_0,u) + 1/w(u,v)$. We here describe the corrected algorithm, which requires a more elaborate proof of correctness.} on the algorithm by D\"urr, Heiligman, H{\o}yer and Mhalla~\cite{durr2006quantum}.
We slightly generalize the algorithm to the case where there is a set of ``forbidden edges'', which we encode by associating a weight $w_e = 0$ to these edges.
The algorithm should then return an SPT over the connected component of the source.
We recall that the cost of traversing an edge is described by its resistance $1/w_e$, and the distance $\delta_G(u,v)$ between nodes $u$ and $v$ is
\[
\delta_G(u,v)
= \min_{u-v \;\mathrm{path}\; P} \sum_{e \in P} \frac{1}{w_e}.
\]
 
The key routine that is used is a quantum algorithm for minimum-finding.
It assumes oracle access to a ``value function'' $f:[N] \to \real \cup \{\infty\}$ and a ``type function'' $g:[N] \to \natural$, with a total number of types
\[
M
= | \mathrm{Im}(g) |
= |\{g(j) \,:\, j\in[N]\}|.
\]
Given an integer $d \in \natural$ with $d \leq N/2$, the problem $\minfind(d,f,g)$ requires to output a subset $I \subseteq [N]$ of size $|I| = \min\{d,M\}$, such that
\begin{itemize}
\item
$g(i) \neq g(j)$ for all distinct $i,j \in I$, and
\item
for all $j \in [N] \backslash I$ and $i \in I$, if $f(j) < f(i)$, then $f(i') \leq f(j)$ for some $i' \in I$ with $g(i') = g(j)$.
\end{itemize}
\begin{proposition}[{\cite[Theorem 3.4]{durr2006quantum}}] \label{prop:min-val}
There is a quantum algorithm to solve $\minfind(d,f,g)$ in time $\tO(\sqrt{Nd} \log(1/\delta))$ with success probability at least $1-\delta$.
The algorithm requires $O(\log N)$ qubits and a QRAM of $\tO(d)$ bits.
\end{proposition}

The quantum algorithm for constructing a shortest-path tree effectively implements Dijkstra's algorithm, which we recall below for completeness.
The algorithm assigns a ``cost'' to any (ordered) pair of nodes $u,v$ by setting
\[
\cost(u,v) = \delta_G(v_0,u) + w(u,v).
\]
We define the \emph{component} of a node $v_0$ as the smallest subset $S \subseteq V$ such that $v_0 \in S$ and either $|E(S,V \backslash S)| = 0$ or $\min\{w(u,v) \mid (u,v) \in E(S,V \backslash S)\} = \infty$.
\begin{algorithm}[H]
\caption{Dijkstra's algorithm}
\label{alg:Dijkstra}
\begin{algorithmic}[1]
\State
let $T = (V_T = \{v_0\}, E_T = \emptyset)$
\While{$|V_T| < n-1$}
\State
find an edge $(u,v) \in E(V_T, V \backslash V_T)$ with minimal $\cost(u,v)$
\If{$w(u,v) = \infty$}
\State
abort and output $T$
\Else
\State
add $v$ to $V_T$ and $(u,v)$ to $E_T$
\EndIf
\EndWhile
\end{algorithmic}
\end{algorithm}

Starting from a node $v_0$, Dijkstra's algorithm (Algorithm \ref{alg:Dijkstra}) returns a shortest-path tree from $v_0$ that spans the component of $v_0$.
A key feature of the algorithm is that, at any time, the tree $T$ is a shortest-path tree for the nodes in $T$ (i.e., $\delta_T(v_0,u) = \delta_G(v_0,u)$ for all $u \in T$).
It follows that for any node $u \in T$ we can efficiently evaluate $\delta_G(v_0,u)$ and hence $\cost(u,v) = \delta_G(v_0,u) + 1/w(u,v)$.

Now we describe a quantum SPT algorithm that is based on Dijkstra's algorithm.
It will have the same feature that at any time, the tree $T$ is a shortest-path tree for the nodes in $T$.
 
\begin{algorithm}[H]
\caption{$T = \spt(G,v_0)$} \label{alg:SPT}
\begin{algorithmic}[1]
\State
let $T = (V_T = \{v_0\},E_T=\emptyset)$, $L = 1$ and $P_1 = \{v_0\}$
\State
set $\dist(v_0) = 0$ and $\dist(u) = \infty$ for all $u \neq v_0$
\While{$|V_T| < n-1$}
\State
let $B_L$ be the output of $\minfind(|P_L|,f,g)$ on $E(P_L)$ where
\begin{itemize}
\item[(i)] type function $g$ returns the end node $g((u,v)) = v$ of an edge $(u,v)$,
\item[(ii)] value function $f((u,v)) = \cost(u,v) = \dist(u) + 1/w(u,v)$ if $u \in P_L, v \notin T$, and $f((u,v)) = \infty$ otherwise
\end{itemize}
\State
let $(u,v) \in B_1 \cup \dots \cup B_L$ have minimal $\cost(u,v)$ with $v \notin P_1 \cup \dots \cup P_L$
\If{$w(u,v) = 0$}
\State
abort and output $T$
\Else
\State
add $v$ to $V_T$ and $(u,v)$ to $E_T$, set $P_{L+1} = \{v\}$ and $L = L+1$
\State
set $\dist(v) = \dist(u) + 1/w(u,v)$
\EndIf
\State
as long as $L \geq 2$ and $|P_L| = |P_{L-1}|$, merge $P_L$ into $P_{L-1}$ (removing $P_L$) and set $L = L-1$
\EndWhile
\end{algorithmic}
\end{algorithm}

First we prove correctness of the algorithm. 
\begin{proposition}
Let $G = (V,E,w)$ be an undirected graph with weights $w: E \to \real_{\geq 0}$, and $C_{v_0}$ the connected component of $v_0 \in V$.
Then Algorithm $\spt(G,v_0)$ outputs a shortest-path tree from $v_0$ that spans $C_{v_0}$.
\end{proposition}
\begin{proof}
We choose the success probability of each call to $\minfind$ to be $1-\delta/n$, so that the total success probability is at least $(1-\delta/n)^n \geq 1-\delta$.
Assuming this, correctness of the algorithm follows from Dijkstra's algorithm provided that we can show that each iteration adds a least-cost border edge to $T$, if it exists.
To this end, notice that in every iteration we have the following invariants directly after step 3:
\begin{enumerate}
\item
The nodes of $T$ are partitioned into subsets $P_1,\dots,P_L$ whose sizes are nonincreasing powers of 2, and such that $|P_k| > \sum_{i=k+1}^L |P_i|$.
\item
For every $1 \leq k \leq L$, the set $B_k$ contains up to $|P_k|$ least-cost border edges from $P_k$ to distinct nodes outside of $P_k$.
\end{enumerate}
Now let $T$ be the SPT constructed so far, let $(u,v)$ be a least-cost border edge of $T$, and let $P_1,\dots,P_L$ denote the node subsets directly after step 3.
Without loss of generality, assume that $u \in P_k$.
Since the edge $(u,v)$ is a least-cost border edge of $T$, it must also be a least-cost border edge of $P_k$ from the subset
\[
E_{\neq k}
= \{(u',v') \mid u' \in P_k, v' \notin V(T) = P_1 \cup \dots \cup P_L\}.
\]
Indeed, if there were another edge $(\hat u,\hat v) \in E_{\neq k}$ with $\cost(\hat u,\hat v) < \cost(u,v)$, then $(\hat u,\hat v)$ would be a border edge of $T$ with lower cost than $(u,v)$, which is a contradiction.

Now it suffices to show that there exists an edge $(u',v') \in B_k$ with $v' \notin T$ and $\cost(u',v') = \cost(u,v)$ (potentially $(u',v') = (u,v)$), because this would imply that step 5 indeed adds a least-cost border edge to $T$.
Recall that $B_k$ was constructed as the set of up to $|P_k|$ least-cost border edges with distinct endpoints from the larger subset\footnote{Note that $B_k$ was constructed at the iteration where $P_k$ was ``formed'' (potentially by merging other subsets), and that the subsets $P_1,\dots,P_k$ do not change between this iteration and the iteration that we are examining.}
\[
E_{> k}
= \{(u',v') \mid u' \in P_k, v' \notin P_1 \cup \dots \cup P_k\}.
\]
By the invariant property we have that $|P_{k+1} \cup \dots \cup P_L| < |P_k|$, and so the set $B_k$ must contain at least one edge $(u',v')$ with $v' \notin P_{k+1} \cup \dots \cup P_L$ (and so $v' \notin T$).
This must be a least-cost edge from $E_{\neq k}$, and so $\cost(u',v') = \cost(u,v)$.
\end{proof}

It remains to prove that there is an efficient quantum algorithm for implementing $\spt$.
\begin{proposition}
There is a quantum algorithm that implements $\spt(G,v_0)$ with success probability $1-\delta$ in time $\tO(\sqrt{|C_{v_0}| |E(C_{v_0})|} \log(n/\delta))$ and by making $\tO(\sqrt{|C_{v_0}| |E(C_{v_0})|} \log(n/\delta))$ quantum queries to the adjacency list of $G$.
The algorithm requires $O(\log n)$ qubits and a QRAM of $\tilde O(|C_{v_0}|)$ bits.
\end{proposition}
\begin{proof}
The bound on the runtime follows from bounding the runtime of the calls to $\minfind$, which are dominant.
It follows from Proposition~\ref{prop:min-val} that a call of $\minfind(|P_L|,f,g)$ requires time $O(\sqrt{|E(P_L)| |P_L|} \log(n/\delta))$.
The total runtime is hence given by
\[
O\bigg( \sum_{P_L} \sqrt{|E(P_L)| |P_L|} \log(n/\delta) \bigg),
\]
where the sum runs over the sets $P_L$ in step 3 of each iteration.
We can bound this by noting that the merging procedure in step 9 ensures that in every iteration $|P_L| = 2^{r_L}$ for some integer $r_L \leq \log n$, and any two $P_L$ of the same size are necessarily disjoint.
Since necessarily $P_L \subseteq C_{v_0}$, there are at most $|C_{v_0}|/2^r$ (disjoint) sets of size $2^r$, and we can bound
\begin{align*}
\sum_{P_L:|P_L| = 2^r} \sqrt{|E(P_L)| |P_L|}
&= 2^{r/2} \sum_{P_L:|P_L| = 2^r} \sqrt{|E(P_L)|} \\
&\leq 2^{r/2} \sqrt{2^{-r} |C_{v_0}| \sum_{P_L:|P_L| = 2^r} |E(P_L)|}
\leq \sqrt{2|C_{v_0}| |E(C_{v_0})|},
\end{align*}
where the second inequality follows from Cauchy-Schwarz and the third inequality from the fact that $\sum_{P_L:|P_L| = 2^r} |E(P_L)| \leq 2|E(C_{v_0})|$.
Summing this over all $r \leq \log n$, we find the claimed runtime.

Finally we bound the space complexity.
The only quantum routine is running $\minfind(|P_L|,f,g)$ on the set of edges, with $|P_L| \leq n$.
By Proposition \ref{prop:min-val} this requires $O(\log n)$ qubits and a QRAM of $\tO(n)$ bits.
All remaining routines are classical operations on the $\tO(n)$ bits describing (i) the nodes and edges of the tree $T$, (ii) the sets $B_1,\dots,B_L$ and (iii) the string $\dist(\cdot)$.
\end{proof}
 
%%%%%%%%%%%%%%%%%%%%%%%%%%%%%%%%%%%%%%%%%%%%%%%%%%%%%%%%%%%%%%%%%%%%%
\section{Existence of Disjoint Matching} \label{app:existence}
%%%%%%%%%%%%%%%%%%%%%%%%%%%%%%%%%%%%%%%%%%%%%%%%%%%%%%%%%%%%%%%%%%%%%
Fix $n$, $m \leq n^2/4$ and $\epsilon \geq \sqrt{n/m}$.
Let $G_1 = (L_1 \cup R_1,E_1)$ consist of $\epsilon^2 n/2$ disjoint copies $B^{(k)}$ of the complete bipartite graph on $1/\epsilon^2$ left and right nodes, containing $2\epsilon^2 m/n$ parallel copies of every edge.
In this way, $G_1$ has $n$ nodes, $m$ edges and is $2m/n$-regular.
We index the $i$-th left and $j$-th right node of $B^{(k)}$ as $l^{(k)}_i$ resp.~$r^{(k)}_j$.
Let $G_2 = (L_2 \cup R_2,E_2)$ be the complete bipartite graph on $2m/n$ left nodes and $n/2$ right nodes.
In this appendix we prove that every edge in $G_1$ can be matched to a unique edge in $G_2$ such that
\begin{enumerate}
\item
all edges leaving a left node $l \in L_1$ are matched to edges in $G_2$ with distinct left ends,
\item
all edges leaving a right node $r \in R_1$ are matched to edges in $G_2$ with distinct right ends.
\end{enumerate}
We do this by considering maximum bipartite matchings in $G_2$ of the form
\[
M_j
= \{(i,i+j) \mid i \in [2m/n]\} \subset E_2,\qquad
0 \leq j < n/2,
\]
where the sum is modulo $n/2$.
These matchings form a partition of the edges set $E_2$.
We interpret every $M_j$ as a set of $2m/n$ ordered edges.
Now for every node $l^{(k)}_i \in L_1$ we will match the $2m/n$ (lexicographically ordered) outgoing edges $E(l^{(k)}_i)$ from that node to the (lexicographically ordered) edges in some $M_j$.
This ensures that all edges leaving $l^{(k)}_i$ are matched to edges in $G_2$ whose left ends are distinct, so that condition 1.~is satisfied.
In order to satisfy condition 2., we specify the matching as follows:
\[
E(l^{(k)}_i)
\Leftrightarrow M_{k - 1 + (i-1)\epsilon^2n/2},\qquad
\forall k \in [\epsilon^2 n/2],\; i \in [1/\epsilon^2].
\]
Indeed, one can check that the $2m/n$ edges matched to for instance the incoming edges of node $r^{(1)}_1$ are described by
\[
(\alpha,\beta),\quad \text{with }
\alpha \in [1,2\epsilon^2m/n] \text{ and } \beta = \ell\frac{\epsilon^2 n}{2} + \alpha,\; 0 \leq \ell < 1/\epsilon^2.
\]
Since we assumed that $m \leq n^2/4$, and hence $2m/n \leq n/2$, the right ends of these edges are indeed disjoint.
The same reasoning applies to all other nodes $r^{(k)}_j$.
We illustrate this matching in Figure \ref{fig:disjoint-matching} below.
 
\begin{figure}[htb]
\centering
\includegraphics[width=.6\textwidth]{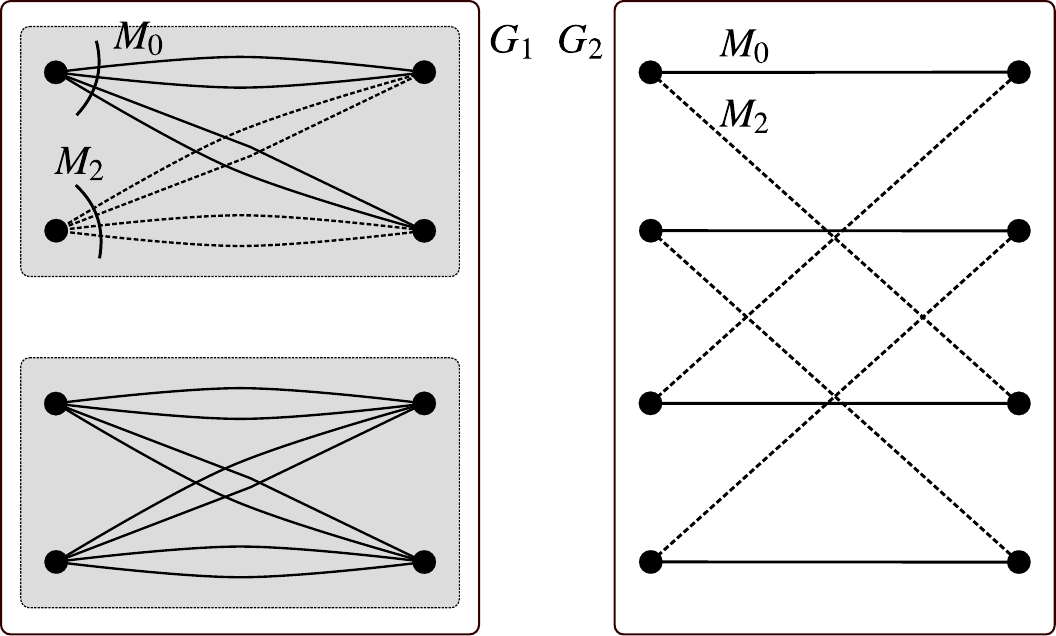}
\caption{Matching edges between $G_1$ and $G_2$ for $n = 8$, $m = 16$ and $\epsilon = 1/\sqrt{2}$. The left ends of matchings $M_0$ are distinct, ensuring condition 1.~for node $l^{(1)}_1$. Similarly, the right ends of all edges matched to $r^{(1)}_1$ are distinct, ensuring condition 2.~for this node.
We note that in general $G_2$ is a bipartite graph between $2m/n$ left nodes and $n/2$ right nodes, with $2m/n \leq n/2$.}
\label{fig:disjoint-matching}
\end{figure}
 
%%%%%%%%%%%%%%%%%%%%%%%%%%%%%%%%%%%%%%%%%%%%%%%%%%%%%%%%%%%%%%%%%%%%%%%%%%%%%%%%%%%%%%%
\section{Quantum SDD Solver} \label{app:SDD}
%%%%%%%%%%%%%%%%%%%%%%%%%%%%%%%%%%%%%%%%%%%%%%%%%%%%%%%%%%%%%%%%%%%%%%%%%%%%%%%%%%%%%%%
 In this appendix we describe a quantum algorithm for approximately solving an SDD linear system $Ax = b$.
The blueprint of the algorithm is given below.
An \textit{SDDM matrix} is an SDD matrix whose off-diagonals are nonpositive.
\begin{algorithm}[H]
\caption{$\mathbf{qSDD}(A,x,b)$} \label{alg:quantum-SDD}
\begin{algorithmic}[1]
\State
reduce the SDD system $Ax = b$ to an SDDM system $\hat A \hat x = \hat b$ using a classical trick \cite{gremban1996combinatorial} (see Section \ref{sec:SDD-to-SDDM})
\State
construct an $\epsilon$-spectral sparsifier $\tilde A$ of $\hat A$ using our quantum sparsification algorithm (see Section \ref{sec:spars-SDDM})
\State
solve the sparse SDDM linear system $\tilde A \hat x = \hat b$ using a \textit{classical} SDD solver \cite{spielman2014nearly}
\end{algorithmic}
\end{algorithm}

By the same reasoning as in Claim \ref{claim:approx-Laplacian}, it follows that an $\epsilon$-approximate solution of the SDDM system $\tilde A \hat x = \hat b$ yields a $2\epsilon$-approximate solution to the system $\hat A \hat x = \hat b$.
This proves correctness of the algorithm.
The runtime is dominated by the second step, which takes time $\tO(\sqrt{mn}/\epsilon)$ (see Section \ref{sec:spars-SDDM}).
Hence this describes an $\epsilon$-approximate quantum SDD solver with runtime $\tO(\sqrt{mn}/\epsilon)$, which proves Proposition \ref{prop:quantum-SDD}.
 
%%%%%%%%%%%%%%%%%%%%%%%%%%%%%%%%%%%%%%%%%%%%%%%%%%%%%%%%%%%%%%%%%%%%%%%%%%%%%%%%%%%%%%%
\subsection{Classical SDD to SDDM Reduction} \label{sec:SDD-to-SDDM}
%%%%%%%%%%%%%%%%%%%%%%%%%%%%%%%%%%%%%%%%%%%%%%%%%%%%%%%%%%%%%%%%%%%%%%%%%%%%%%%%%%%%%%%
Gremban \cite{gremban1996combinatorial} showed how to reduce an SDD system $Ax = b$ to an SDDM system $\hat{A} \hat{x} = \hat{b}$.
Let $A = D + P + N$ be an SDD matrix, with the diagonal matrix $D$ containing the diagonal elements, and $N$ and $P$ containing the negative and positive off-diagonal elements, respectively.
To this matrix we can associate an SDDM matrix $\hat{A} \in \real^{2n \times 2n}$ and a vector $\hat{b} \in \real^{2n \times 1}$, which we define as
\[
\hat{A}
= \begin{bmatrix} D + N & -P \\ -P & D + N \end{bmatrix}, \qquad
\hat{b}
= \begin{bmatrix} b \\ -b \end{bmatrix}.
\]
Now Gremban showed that if the $2n$-dimensional vector $\hat{z} = [ \hat{z}_1 \hat{z}_2]^T$ is an $\epsilon$-approximate solution of the linear system $\hat{A}\hat{x} = \hat{b}$, then the $n$-dimensional vector $z = (\hat{z}_1 - \hat{z}_2)/2$ is an $\epsilon$-approximate solution of the original system $Ax = b$.
More precisely:
\[
\| \hat{z} - \hat{A}^+ \hat{b} \|_{\hat A} \leq \epsilon \| \hat{A}^+ \hat{b} \|_{\hat A}
\qquad \Rightarrow \qquad
\| z - A^+ b \|_A \leq \epsilon \| A^+ b \|_A.
\]
This implies that it suffices to solve the SDDM system $\hat{A} \hat{x} = \hat{b}$.
 
Given sparse access to $A$, we now show how to simulate sparse access to $\hat{A}$.
The latter takes as input a pair $(i,k) \in [2n]^2$.
If $i \leq n$, let $(j,A_{ij})$ denote the output of an $(i,k)$-query to the original matrix $A$.
If this yields an error symbol, output the error symbol.
Otherwise, do the following:
\[
(i,k)
\mapsto
\begin{cases}
(j,A_{ij}) &\text{if } j = i \text{ or } A_{ij}<0 \\
(j+n,-A_{ij}) &\text{otherwise.}
\end{cases}
\]
If $i \geq n+1$, let $(j,A_{ij})$ denote the output of an  $(i-n,k)$-query to $A$.
If this yields an error symbol, return the error symbol, otherwise do the following:
\[
(i,k)
\mapsto
\begin{cases}
(j+n,A_{ij}) &\text{if } j = i-n \text{ or } A_{ij}<0 \\
(j,-A_{ij}) &\text{otherwise.}
\end{cases}
\]
This proves that we can simulate sparse access to $\hat{A}$ using a single query to $A$.
Note that the relative ordering between elements in the induced adjacency list for $\hat{A}$ may be different from the ordering in the adjacency list for $A$; that is not a problem because we are not assuming the adjacency lists to be ordered in any particular way.
 
%%%%%%%%%%%%%%%%%%%%%%%%%%%%%%%%%%%%%%%%%%%%%%%%%%%%%%%%%%%%%%%%%%%%%%%%%%%%%%%%%%%%%%%
\subsection{Sparsifying SDDM Matrices} \label{sec:spars-SDDM}
%%%%%%%%%%%%%%%%%%%%%%%%%%%%%%%%%%%%%%%%%%%%%%%%%%%%%%%%%%%%%%%%%%%%%%%%%%%%%%%%%%%%%%%
By the previous section, it suffices to solve an SDDM system $\hat A \hat x = \hat b$ to which we have sparse access.
We will do so by using our quantum algorithm to construct a much sparser SDDM matrix $\tilde A$ with respect to which we can solve linear systems classically.
If we can ensure that\footnote{We use the shorthand $B \approx_\epsilon C$ to mean that $(1-\epsilon) B \preceq C \preceq (1+\epsilon) \hat B$.}
\[
\tilde A \approx_\epsilon \hat A,
\]
then an $\epsilon$-approximate solution of $\tilde A \hat x = \hat b$ will be an $O(\epsilon)$-approximate solution to $\hat A \hat x = \hat b$.
 
Since $\hat A$ is SDDM, we can rewrite it as $\hat A = \hat L + \hat D$, with $\hat L$ a Laplacian and $\hat D = \mathrm{diag}(\hat A) - \mathrm{diag}(\hat L)$ a nonnegative, diagonal ``excess''-matrix.
Given sparse access to $\hat A$, we can easily simulate adjacency list access to the graph associated to $\hat L$.
Using our sparsification algorithm, this allows to explicitly output a sparsified Laplacian $\tilde L$ with $\tO(n/\epsilon^2)$ nonzero entries such that
\[
\tilde L
\approx_\epsilon \hat L.
\]
Now we define the matrix
\[
\tilde A
= \tilde L + \mathrm{diag}(\hat A) - \mathrm{diag}(\tilde L).
\]
This matrix has $\tO(n/\epsilon^2)$ nonzero entries, of which we have an explicit description.
We can prove that $\tilde A$ closely approximates $\hat A$, as follows.
\begin{lemma}
Let $\tilde A$ be as described above.
Then $\tilde A \approx_{2\epsilon} \hat A$.
\end{lemma}
\begin{proof}
First we use the fact that if $A$, $B$ and $C$ are PSD matrices and $A \approx_\epsilon B$, then $A + C \approx_\epsilon B + C$.
Since $\hat L \approx_\epsilon \tilde L$, this shows that $\hat A = \hat L + \hat D \approx_\epsilon \tilde L + \hat D$.
As a special case, this implies that $\mathrm{diag}(\hat A) = \mathrm{diag}(\hat L) + \hat D \approx_\epsilon \mathrm{diag}(\tilde L) + \hat D$.
Now we can simply rewrite
\[
\tilde A
= \tilde L + \hat D + \mathrm{diag}(\hat L) - \mathrm{diag}(\tilde L)
\approx_\epsilon \tilde L + \hat D + \mathrm{diag}(\tilde L) - \mathrm{diag}(\tilde L)
= \tilde L + \hat D
\approx_\epsilon \hat A.
\]
Now it suffices to note that if $\tilde A \approx_\epsilon B \approx_\epsilon \hat A$ then $\tilde A \approx_{2\epsilon} \hat A$.
\end{proof}
 
We can hence solve the sparse linear system $\tilde A \hat x = \hat b$ to yield an approximation of the original system.
Since $\tilde A$ has $\tO(n/\epsilon^2)$ nonzero entries, we can do this efficiently using a classical near-linear time SDD solver \cite{spielman2014nearly}.

\end{document}